\documentclass[a4paper]{article}
\usepackage[utf8]{inputenc}
\usepackage{algorithm, algpseudocode, amsmath, amssymb, amsthm,amsfonts,mathrsfs}
\usepackage{ae,aecompl,array}
\usepackage{color,graphics,hyperref,cleveref, tikz}
\usepackage{dsfont}
\usepackage{thm-restate}
\usepackage{bbm,verbatim }
\usepackage{stmaryrd}
\usepackage{soul}
\usepackage{cancel}

\setstcolor{red} \usepackage[margin=1.314in]{geometry}

\usepackage{tikz}

\hypersetup{
	colorlinks,
	linkcolor={red!70!black},
	citecolor={blue},
	urlcolor={blue}
}

\usepackage[utf8]{inputenc}

\newcommand{\hf}{\widehat{f}}

\newcommand{\tp}{\widetilde{p}}
\newcommand{\COMMENT}[1]{}

\newcommand{\F}{\mathbb{F}}
\newcommand{\Fq}{{\F_q}}
\newcommand{\Fn}[2]{\F_{#1}^{#2}}

\newcommand{\Fqn}{\Fn{q}{n}}
\newcommand{\Fqni}{\Fn{q}{n_i}}
\newcommand{\Fqk}{\Fn{q}{k}}

\newcommand{\N}{\mathbb{N}}

\newcommand{\Real}{\mathbb{R}}
\newcommand{\Comp}{\mathbb{C} }

\newcommand{\D}{\mathcal{D}}

\DeclareMathOperator{\tr}{tr}

\newcommand{\trsp}[1]{{#1}^{\intercal}}

\newcommand{\eqdef}{\mathop{=}\limits^{\triangle}}
\newcommand{\Unif}{\leftarrow}

\newcommand{\ie}{\textit{i.e.}}
\newcommand{\eps}{\epsilon}
\newcommand{\rk}{\text{rk}}

\newcommand{\PPGM}{\textrm{P}_{\textrm{PGM}}}

\newcommand{\esp}{\mathbb{E}}
\newcommand{\psucc}{{p_\textrm{succ}}}
\DeclareMathOperator{\Var}{Var}

\DeclareMathOperator{\negl}{negl}

\newcommand{\ket}[1]{|#1\rangle}
\newcommand{\bra}[1]{\langle#1|}
\newcommand{\ketbra}[2]{|#1\rangle\langle#2|}
\newcommand{\braket}[2]{\langle #1 | #2 \rangle}
\newcommand{\altketbra}[1]{\ketbra{#1}{#1}}
\newcommand{\kb}[1]{\altketbra{#1}}
\newcommand{\norm}[1]{\ensuremath{\lvert  \kern-1pt\lvert #1 \rvert \kern-1pt \rvert}}
\newcommand{\pnorm}[2]{\ensuremath{\lvert  \kern-1pt\lvert #1 \rvert \kern-1pt \rvert_{#2}}}
\newcommand{\abs}[1]{|#1|}

\newcommand{\QDP}{\ensuremath{\mathrm{QDP}}}
\newcommand{\DP}{\mathrm{DP}}

\newcommand{\LWE}{\ensuremath{\mathrm{LWE}}}
\newcommand{\SLWE}{\ensuremath{\mathrm{S-}\ket{\mathrm{LWE}}}}
\newcommand{\SCP}{\mathrm{SCP}}

\DeclareMathOperator{\QFTt}{QFT}
\newcommand{\QFT}[1]{\widehat{#1}}

\newcommand{\ff}{\QFT{f}}
\newcommand{\fg}{\QFT{g}}

\newcommand{\carf}[1]{\chi_{#1}}
\newcommand{\car}[2]{\chi_{#1}\left(#2 \right)}

\newcommand{\Rbra}[1]{\left( #1 \right)}

\newcommand{\mat}[1]{\ensuremath{\boldsymbol{#1}}}

\newcommand{\cv}{{\mat{c}}}
\newcommand{\dv}{\mat{d}}
\newcommand{\ev}{\mat{e}}

\newcommand{\sv}{{\mat{s}}}

\newcommand{\uv}{\mat{u}}
\newcommand{\vv}{\mat{v}}
\newcommand{\xv}{{\mat{x}}}
\newcommand{\yv}{{\mat{y}}}
\newcommand{\zv}{{\mat{z}}}
\newcommand{\zero}{{\mat{0}}}

\newcommand{\Gm}{\ensuremath{\mathbf{G}}}
\newcommand{\Hm}{\ensuremath{\mathbf{H}}}
\renewcommand{\Im}{\ensuremath{\mathbf{I}}}
\newcommand{\Mm}{\ensuremath{\mathbf{M}}}
\newcommand{\Xm}{\ensuremath{\mathbf{X}}}

\DeclareMathOperator{\rank}{rank}
\DeclareMathOperator{\matv}{Mat}

\newcommand{\tW}{\widetilde{W}}

\newcommand{\wpsi}{\widetilde{\psi}}
\newcommand{\wphi}{\widetilde{\phi}}

\newcommand{\OO}[1]{O\Rbra{#1}}

\newcommand{\Th}[1]{\Theta\Rbra{#1}}

\newcommand{\Ac}{{\mathcal{A}}}
\newcommand{\Bc}{{\mathcal{B}}}
\newcommand{\C}{{\mathscr{C}}}

\newcommand{\Ec}{{\mathcal{E}}}

\newcommand{\Hc}{{\mathcal{H}}}

\newcommand{\Tc}{{\mathcal{T}}}

\newcommand{\CC}{\ensuremath{\mathscr{C}}}
\newcommand{\CCp}{{\ensuremath{\CC^\perp}}}

\newcommand{\CCps}{{\ensuremath{(\CCp)^*}}}
\newcommand{\Typ}{{\ensuremath{\Tc_\epsilon^{(n)}}}}
\newcommand{\Typi}{{\ensuremath{\Tc_\epsilon^{(n_i)}}}}
\newcommand{\Typii}{{\ensuremath{\Tc_{\epsilon(i)}^{(n_i)}}}}
\newcommand{\Typc}[1]{{\ensuremath{\Tc_\epsilon^{(#1)}}}}
\newcommand{\Typcc}[2]{{\ensuremath{\Tc_{#2}^{(#1)}}}}
\newcommand{\bTyp}{{\ensuremath{\Bc_\epsilon^{(n)}}}}

\newcommand{\sE}{{\mathscr{E}}}

\newcommand{\sS}{{\mathscr{S}}}
\newcommand{\sT}{{\mathscr{T}}}

\newtheorem{theorem}{Theorem}

\newtheorem{fact}{Fact}
\newtheorem{definition}{Definition}
\newtheorem{lemma}{Lemma}
\newtheorem{proposition}{Proposition}
\newtheorem{corollary}{Corollary}
\newtheorem{claim}{Claim}

\newtheorem{remk}{Remark}

\newtheorem{problem}{Problem}

\renewcommand{\and}{\mbox{ and }}

\newcommand\lp{\left (}
\newcommand\rp{\right )}
\newcommand\lc{\left [}
\newcommand\rc{\right ]}

\title{The Quantum Decoding Problem : Tight Achievability Bounds and Application to Regev's Reduction}
\author{Agathe Blanvillain$^{1}$ and Andr\'e Chailloux$^{1}$ and Jean-Pierre Tillich$^{1}$\\
Inria de Paris, {\sf $\{$agathe.blanvillain,andre.chailloux,jean-pierre.tillich$\}$@inria.fr}}
\date{}
\begin{document}
	\maketitle
	\begin{abstract}

	We consider the quantum decoding problem. It consists in recovering a codeword given a superposition of noisy versions of this codeword. By measuring the superposition, we get back to the classical decoding problem. It appears for the first time in Chen, Liu and Zhandry's work showing a quantum advantage for the Short Integer Solution (SIS) problem for the $l_\infty$ norm. In a recent paper, Chailloux and Tillich proved that when we have a noise following a Bernoulli distribution, the quantum decoding problem can be solved in polynomial time and is therefore easier than classical decoding for which the best known algorithms have an exponential complexity. They also give an information theoretic limit for the code rate at which this problem can be solved which turns out to be above the Shannon limit.
	
	In this paper, we generalize the last result to all memoryless noise models. We also show similar results in the rank metric case which corresponds to a noise model which is not memoryless. We analyze the Pretty Good Measurement, from which we derive an information theoretic limit for this problem. By using the algorithm for the quantum decoding problem together with Regev's reduction, we derive a quantum algorithm sampling codewords from the dual code according to a probability distribution which is the dual of the original noise. It turns out that at the information theoretic limit, we get the most likely nonzero codeword of the dual code. When the distribution is a decreasing function of the weight, we find minimal nonzero codewords. Note that Regev’s reduction used together with classical decoding is much less satisfying since it is not able to output those minimum weight codewords.

	\end{abstract}
	\newpage
	\tableofcontents
	\newpage
	\section{Introduction}
	\subsection{State of the art}
\paragraph{Post quantum cryptography.}
Post-quantum cryptography's aim is to deploy cryptographic schemes that are secure against quantum computers. Some of the main proposals for post-quantum cryptography are based on variants of decoding problems, be they lattice based or code based. This is clear in code-based cryptography where the overwhelming majority of proposals is  based on decoding either in the Hamming or in the rank metric. In lattice based cryptography, the Learning With Errors ($\LWE$) problem \cite{R05} is ubiquitous and is nothing but decoding a linear code over a large alphabet subject to a
discrete Gaussian noise.
Most of quantum algorithms \cite{B10,KT17a,L16,CL21} for attacking these post-quantum cryptographic schemes  are quantum extensions of the best classical algorithms, using simple quantum algorithmic tools such as Grover's algorithm~\cite{Gro96}, or more advanced tools such as quantum walks~\cite{AAKZ01,MNRS11}. While these algorithms are important for assessing the security of post-quantum cryptographic schemes, they intrinsically fail to provide exponential speedups or indications of strong quantum advantage.

\paragraph{Regev's reduction.} Recently, a new family of quantum algorithms was developed inspired by Regev's reduction~\cite{R05}. Originally, Regev's reduction is a quantum complexity reduction from the Learning with Errors problem to the Short Integer Solution (SIS) problem. More precisely, a quantum algorithm for the SIS problem is devised from a (classical or quantum) algorithm for the LWE problem. This can be rephrased in coding theoretic terms. With the help of a decoding algorithm operating in the regime where there is generally at most a unique solution, an algorithm outputting a rather short dual codeword is obtained, where short is measured in terms of the Euclidean distance. Here the code $\C$ which is decoded is a generic linear code and so is the dual code $\C^\bot$, where we recall that
$$ \C = \{\sv \Gm : \sv \in \F_q^k\} \subseteq \F_q^n \quad ; \quad \C^\bot = \{\yv \in \F_q^n : \forall \cv \in \C, \ \yv \cdot \cv = 0\}.$$

More precisely Regev's algorithm reduces the short codeword problem 

\begin{problem}[Short Codeword Problem $\SCP(q, n, k, w)$]\mbox{ }\\
	Take $q, n, k, w$ positive integers. $\Hm$ is chosen uniformly at random in $\F_q^{(n-k) \times n}$. \\
	{\bf Given: }$\Hm$, \\
	{\bf Goal:} find $\cv \in \Fqn\backslash\lbrace 0 \rbrace$ which satisfies $\Hm\trsp{\cv} = 0$ and $\abs{\cv} \leq w$.
\end{problem}
to the decoding problem described by
\begin{problem}[Decoding Problem $\DP(q, n, k, p)$]
	Take $q, n, k$ positive integers and $p$ a probability distribution over $\F_q^n$.
	Take a generator matrix $\Gm$ in  $\F_q^{k\times n}$, a codeword $\cv =\uv\Gm$ of the code $\C$ generated by $\Gm$ where 
	$\uv$ is chosen uniformly at random in $\Fqk$. \\
	{\bf Given:}   $(\Gm, \cv+\ev)$, where $\ev$ is sampled according to $p$, \\
	{\bf Goal:} recover $\cv$.
\end{problem}

Here $\abs{\cv}$ is the weight of the codeword $\cv = c_1,\dots,c_n$. In Regev's reduction it is the Euclidean weight $|\cv| = \sqrt{\sum_{i=1}^n c_i^2}$ where the sum is performed over the integers and $c_i$ is taken in the integer interval $\left\{- \frac{q-1}{2},\cdots,\frac{q-1}{2} \right\}$ (in the case where $q$ is odd and $\left\{- \frac{q}{2},\cdots,\frac{q-2}{2} \right\}$ otherwise). Note that in both cases we have a random linear code that we either want to decode or for which we  look for a short codeword. The originality of this reduction is that it uses a decoding algorithm which works with a decoding radius where the solution is unique (this can be viewed as the injective regime) to get a short codeword for a weight $w$ where there are exponentially many codewords of this weight (this can be viewed as the surjective regime). It took a while to get an analogue of this reduction for the Hamming metric \cite{DRT24} where to get a non trivial reduction, it is mandatory
in certain parameter regimes to consider a decoding algorithm which works for a larger decoding radius where the solution of the decoding algorithm is typically unique but not always unique.

Regev's reduction makes crucially use of the quantum Fourier transform (QFT). It maps the quantum state $\sum_{\xv \in \F_q^n} f(\xv) \ket{\xv}$ to $\sum_{\xv \in \F_q^n} \hf(\xv) \ket{\xv}$ where $\hf$ is the classical Fourier transform of $f$. The latter is defined by using
the characters $\chi_\yv$ of $\F_q^n$ (see \S \ref{ss:notation}), the Fourier transform over $\Fqn$ being defined by $\hat{f}(\yv) \eqdef \frac{1}{\sqrt{q^n}} \sum_{\xv \in \Fqn}
\chi_{\yv}(\xv) f(\xv)$. For this we first fix a function $f : \F_q^n \rightarrow \mathbb{C}$ with $\norm{f}_2 = 1$ so that $|f|^2$ will serve as a probability distribution of the errors we feed in the decoder. Regev's reduction consists in performing the following operations:
\begin{align*}
	&\text{Initial state preparation:} &  & \quad \frac{1}{\sqrt{\abs{\C}}}\sum_{\cv \in \C}\sum_{\ev \in \F_q^n}f(\ev)\ket{\cv}\ket{\ev} \\
  &\text{adding $\cv$ to $\ev$:} & \mapsto &  \quad \frac{1}{\sqrt{|\C|}} \sum_{\cv \in \C} \sum_{\ev \in \F_q^n}  f(\ev) \ket{\cv}\ket{\cv+\ev}  \\
   &\!\!\!\begin{array}{l} \text{recovering $\cv$ from decoding $\cv+\ev$ and }\\
    \text{substracting it to the 1st register to erase it} \end{array} & \mapsto  & \quad 
  \ket{\zero} \frac{1}{\sqrt{|\C|}} \sum_{\cv \in \C, \ev \in \F_q^n} f(\ev) \ket{\cv + \ev} \\
&\text{apply the QFT:} & \mapsto & \quad \frac{1}{\sqrt{Z}} \sum_{\dv \in \C^\perp} \hat{f}(\dv) \ket{\dv}  \\
&\text{measuring the whole state:} & \mapsto & \quad \ket{\dv} \;\;\text{with    $\dv \in \C^\perp$.}
\end{align*}
The idea of Regev's algorithm is to make use of a classical decoder recovering $\cv$ (and therefore $\ev$) from $\cv + \ev$.
To be able to perform this operation, we are necessarily in the aforementioned injective regime of decoding. Such an algorithm, if it exists, allows us to ``erase" the $\ket{\cv}$ register when applied coherently.
The resulting state $\frac{1}{\sqrt{|\C|}} \sum_{\cv \in \C, \ev \in \F_q^n} f(\ev) \ket{\cv + \ev}=\sum_{\yv \in \F_q^n} \alpha_\yv \ket{\yv}$ is periodic in the sense that its amplitudes $\alpha_\yv$ clearly satisfy $\alpha_{\yv} = \alpha_{\yv+\cv}$ for any $\cv$ in $\C$. This explains why when we apply the
quantum Fourier transform to it, we only get non zero amplitudes on the orthogonal space $\C^\bot$. Measuring the final state will give a dual codeword according to a probability distribution proportional to $|\hf|^2$ restricted on the dual code. If $|f(\ev)|^2$ concentrates its weight on small $\ev$'s, then it turns out that in many examples of interest $|\hat{f}(\dv)|^2$ also concentrates on rather small $\dv$'s. This gives a way to sample
codewords in $\C^\bot$ of rather small weight/norm.

It is worthwhile to note that by using the same decoding algorithm with a slight twist in Regev's reduction, we can  construct $\sum_{\yv  \in \vv+\C^\perp} \hf(\yv) \ket{\yv}$ for any $\vv \in \F_q^n$. Indeed, Regev's reduction is slightly modified by starting instead of the state $\frac{1}{\sqrt{|\C|}} \sum_{\cv \in \C, \ev \in \F_q^n} f(\ev) \ket{\cv}\ket{\cv + \ev}$ by a slight twist on the state obtained by tweaking the amplitudes with characters as follows $\frac{1}{\sqrt{|\C|}} \sum_{\cv \in \C, \ev \in \F_q^n} \chi_\cv(-\vv)f(\ev) \ket{\cv}\ket{\cv + \ev}$. From this state, we perform now the same set of operations as in Regev's reduction:
\begin{align}
  &\frac{1}{\sqrt{|\C|}} \sum_{\cv \in \C, \ev \in \F_q^n} \chi_\cv(-\vv)f(\ev) \ket{\cv}\ket{\cv + \ev}  \\ \xrightarrow{\textrm{Decoder}} &\frac{1}{\sqrt{|\C|}} \sum_{\cv \in \C, \ev \in \F_q^n} \chi_\cv(-\vv)f(\ev) \ket{\cv + \ev}\\  \xrightarrow{\QFTt} &\frac{q^n}{\sqrt{|\C|}} \sum_{\yv \in \vv+\C^\perp} \hf(\yv)\ket{\yv}. \label{eq:VarRegev}
\end{align}
By measuring the resulting state $\sum_{\yv  \in \vv+\C^\perp} \hf(\yv) \ket{\yv}$ we obtain an element in $\vv + \C^\perp$ such that
$\yv$ is relatively small. In other words, we solve the variant of the decoding problem for the dual code $\C^\perp$ for which we just have to find for any word $\vv$ of the ambient space
a sufficiently close codeword. This is given here by
\begin{problem}[Close Codeword Finding in the surjective regime (or finding a close enough codeword) $\DP(q,n,k,w)$]
  $q, n, k, w$ are positive integers such that $k< n$. $\Hm$ is chosen uniformly at random in $\F_q^{(n-k) \times n}$ and $\yv$ uniformly at random in
  $\F_q^n$\\
	{\bf Given: }$\Hm$, $\yv$,\\
	{\bf Goal:} find $\cv \in \Fqn\backslash\lbrace 0 \rbrace$ which satisfies $\Hm\trsp{\cv} = 0$ and $\abs{\cv-\yv} \leq w$.
\end{problem}
It generally appears like that in rate distortion theory, but in this case the code is generally not chosen uniformly at random but in a family where the problem of finding a close enough codeword can be easily solved usch as for instance convolutional codes, polar codes or LDGM codes. It might be also called ``decoding in the surjective regime'', where we are not asked to find the closest codeword, but a close enough codeword. In other words, Regev's reduction can be tweaked to transform a classical decoding algorithm working in the injective regime into a quantum decoding algorithm working in the surjective regime. Recently this approach has been applied on structured code families such as Reed-Solomon codes
\cite{JSW+25,CT25a} for which there exist efficient
decoding algorithms in the injective regime but not in the surjective regime. 
This approach yields a polynomial time algorithm for decoding the dual Reed Solomon code (which is also a Reed-Solomon code in that case since it is chosen as having full support) in the surjective regime. For suitable choices of the function $f$ this seems
to yield a quantum polynomial time algorithm for decoding in cases where the best known classical algorithms take exponential time, thus yielding a problem which is a good candidate for quantum advantage.

\paragraph{Improving Regev's reduction by quantum decoding.}
In principle, this reduction gives a quantum algorithm for producing in a code low weight codewords. As a rule of thumb, the greater the decoding radius is for the decoding algorithm used for
erasing the first register, the lower the weight of the codewords is after measuring in the last step. However, this reduction is somewhat unsatisfactory in the sense that even if we take the
best decoding algorithm for the Hamming distance, we are far for obtaining the minimum (Hamming) weight codewords in the dual code \cite{DRT24}. There has been a recent trend to improve this reduction
by a slight variation of the step consisting in erasing the first register \cite{CLZ22}. Indeed, instead of trying to recover $c$ from $c+e$, the new approach is to observe that the state at the second step can be written as
$$
\frac{1}{\sqrt{|\C|}} \sum_{\cv \in \C} \sum_{\ev \in \F_q^n}  f(\ev) \ket{\cv}\ket{\cv+\ev} = \frac{1}{\sqrt{|\C|}} \sum_{\cv \in \C} \ket{\cv}\ket{\psi_\cv},\;\text{ where }
\ket{\psi_{\cv}} \eqdef \sum_{\ev \in \F_q^n}  f(\ev)\ket{\cv+\ev}.
$$
$\ket{\psi_{\cv}}$ can be viewed as a superposition of all noisy versions of the codeword $\cv$ and to recover $\cv$ and then subtract it to erase the first register, we might as well recover
$\cv$ from this noisy superposition $\ket{\psi_{\cv}} = \sum_{\ev \in \F_q^n}  f(\ev)\ket{\cv+\ev}$. In other words we have to solve the following {\em quantum decoding problem}
(which is called $\SLWE$ in \cite{CLZ22})
\begin{problem}[quantum decoding problem $\text{QDP}(q, n, k, f)$]
	Take $q, n, k$ positive integers and $f : \F_q^n \rightarrow \mathbb{C}$ so that $\norm{f}_2 = 1$.
	
	Take $\Gm \Unif \F_q^{k\times n}$, $\uv \Unif \Fqk$, and $\cv = \uv\Gm$, $\ket{\psi_\cv} = \sum_{\ev \in \Fqn}f(\ev)\ket{\cv+\ev}$.
	
	Given $(\Gm, \ket{\psi_\cv})$, the goal is to recover $\cv$.
\end{problem}
Note that this problem reduces to the classical decoding problem $\DP(q,n,k,p)$ where $p(\xv) = |f(\xv)|^2$ by measuring the quantum input state. However, this quantum decoding problem might turn out to be intrinsically easier than
classical decoding. Indeed, for an $f$ which has a product structure $f(\ev)= \prod_{i=1}^n g(e_i)$, Chen, Liu and Zhandry showed in \cite{CLZ22} that one can use a well chosen quantum measurement on each qudit of $\ket{\psi_\cv}$ (note that in this case $\ket{\psi_\cv}$ has the form
$\bigotimes_{i=1}^n \sum_{e_i \in \Fq} g(e_i) \ket{c_i + e_i}$) yielding partial information on the symbols $c_i$ of the codeword $\cv$. This in turn can be used together with the Arora-Ge algorithm
\cite{AG11} to recover $\cv$. With a well chosen function $g$ which is such that its Fourier transform $\hat{g}$ has support included in $\{-\frac{q-1}{2}+c,\cdots,\frac{q-1}{2}-c\}$ for some constant $c$, one obtains
after performing the rest of Regev's reduction (namely applying the QFT after erasing the first register and measuring at the end) codewords $\yv$ in $\C^\perp$ which are of infinite norm
$\norm{\yv}_\infty \leq \frac{q-1}{2}-c$. It is then shown in~\cite{CLZ22} that this results for constant $c$ in a quantum polynomial time algorithm performing this task of producing rather low weight codewords
for the infinite norm in a regime of parameters for which all classical algorithms known at that time had subexponential complexity. In other words, this gave an excellent candidate for providing a quantum advantage. This research thread has been further explored in \cite{CHLLT23,DFS24}. Note that very recently a new classical algorithm was devised \cite{KOW25} for finding low weight codewords for the infinite norm. It outputs in polynomial a solution of this problem for a slightly larger set of parameters for which the polynomial quantum algorithm of \cite{CLZ22} works, thus showing that this quantum algorithm does not yield a quantum speedup anymore. In any case, this whole thread of research has definitely proved to be fruitful for improving significantly the state of the art for finding short codewords for the infinite norm.

Understanding in depth the quantum decoding problem is therefore a very important question. 
Probably the most fundamental  issue in this respect is to ask is for a given family $(f_n)_{n \geq 1}$ of noise functions what is the largest achievable
 rate $R$ for which as soon as the rate $\frac{k}{n}$  is below $R$ there exists a  quantum algorithm (of potentially unbounded complexity) that solves this problem with high probability. The authors of~\cite{CT24} give the 
largest achievable rate in the case $f$ corresponds to a $q$-ary symmetric channel, namely when
$f(\ev) = \sqrt{(1-p)^{n - |\ev|}\left(\frac{p}{q-1}\right)^{|\ev|}}$ for some $p$ which is the crossover probability of the $q$-ary symmetric channel.
In this work, we generalize the above result to any product error function $f_n = g^{\otimes n}$ giving generic achievability and matching non-achievability bounds. Contrarily to the non-achievability results of \cite{CT24} which hold only for  random codes, we prove the relevant non-achievability result for {\em any code}.  While our work share some techniques with~\cite{CT24}, we use more information theoretic tools which allow us to cleanly phrase our results in terms of the $q$-ary entropy of the function $|\hf|^2$. We now present our contributions more concretely.

\subsection{Contributions}
\paragraph{The capacity of the classical-quantum channel.}
The mapping 
\begin{eqnarray*}
\alpha \in \Fq &\rightarrow& \ket{\psi_\alpha} \eqdef \sum_{u \in \F_q} g(u) \ket{\alpha + u}\\
\cv = (c_1,\cdots,c_n) & \rightarrow & \ket{\psi_{\cv}} = \ket{\psi_{c_1}}\otimes \cdots \otimes \ket{\psi_{c_n}}
\end{eqnarray*}
 for $g$ of norm $1$ can be considered as a classical-quantum (c-q) pure state channel \cite[\S 5.1]{H12}. The maximum rate at which we can transmit information on such a channel is given by the Holevo capacity $C_\chi$ of the channel (see \cite[Theorem 5.4]{H12}) which in the case of our pure state channel takes the form
\begin{equation}\label{eq:Holevo} C_\chi = \sup_{p_\alpha} S\left(\sum_{\alpha \in \F_q} p_\alpha \rho_\alpha \right)
\end{equation}
 where $p_\alpha$ is a probability distribution over $\F_q$, $\rho_\alpha \eqdef \ketbra{\psi_\alpha}{\psi_\alpha}$ and 
$S(\rho)$ is the von Neumann entropy of the quantum state $\rho$ (see \S \ref{ss:quantum_information}). It is readily seen (see \S \ref{ss:preliminaries}) that the Holevo capacity of this c-q channel is given by 
\begin{equation}
\label{eq:capacity}
C_\chi = H_q(|\fg|^2),
\end{equation}
where $H_q$ is the $q$-ary entropy (see Definition~\ref{def:q-entropy} for the definition of the $q$-ary entropy $H_q$ if needed).
This follows on the spot by noticing that the capacity of this c-q channel is the same as the one that maps $\alpha$ to the quantum Fourier transform $\QFT{\ket{\psi_\alpha}}$ of $\ket{\psi_\alpha}$. It is easy to see that
the Holevo capacity of this channel is nothing but $S\left(\frac{1}{q}\sum_{\alpha \in \F_q} \rho'_\alpha\right)$ where $\rho'_\alpha \eqdef \ketbra{\QFT{\psi_\alpha}}{\QFT{\psi_\alpha}}$. The matrix $\frac{1}{q}\sum_{\alpha \in \F_q} \rho'_\alpha$ is nothing but the diagonal matrix with entries the $|\fg(\alpha)|^2$. This explains \eqref{eq:capacity}.
Therefore, $H_q(|\fg|^2)$ is an upper-bound on the achievability rate for our decoding problem, since the issue in our case is whether or not random linear codes achieve the Holevo capacity. 

\paragraph{Achievability result for random linear codes.}
This turns out to be the case and the fact that random linear codes achieve the Holevo capacity of this c-q channel basically comes from the fact that the supremum in \eqref{eq:Holevo} is 
attained for a uniform distribution $p_\alpha$, \ie\  the Holevo capacity $C_\chi$ of this c-q channel is nothing but the symmetric Holevo capacity of the channel that is
\begin{eqnarray*}
C_\chi & = & S(\rho),\;\text{where}\\
\rho & \eqdef & \frac{1}{q} \sum_{\alpha \in \F_q} \rho_\alpha.
\end{eqnarray*}
It turns out in our case that the optimal decoding strategy for linear codes (meaning here the one that maximizes the average probability of success) is given by the Pretty Good Measurement (PGM) which is a generic measurement for quantum state discrimination. This measurement was already studied in~\cite{CT24} in the case $f$ corresponded to a $q$-ary symmetric channel. We generalize this result to treat the case of a general function $f$ which shows that the Holevo capacity of this channel is indeed attained by linear codes
\begin{restatable}{theorem}{tractability}\label{Theorem:Tractability}\label{thm:achievability}
	Let $q \ge 2$ be a prime power, $R \in (0,1)$ and $g : \F_q \rightarrow \mathbb{C}$ such that $\norm{g}_2 = 1$. If $R < H_q(|\fg|^2)$ then there exists a quantum algorithm that solves $\QDP(q,n,\lfloor Rn \rfloor, g^{\otimes n})$ with probability $1 - \negl(n)$.
\end{restatable} 
This theorem applies to a memoryless channel, but it turns out that a more general result is true which applies to more general noise models, namely Theorem \ref{thm:general} and Theorem \ref{thm:achievability}
is only a corollary of this more general result.

\paragraph{Non-achievability result in the general case.}
The fact that the Holevo capacity of the c-q channel is equal to $H_q(|\fg|^2)$ already shows  that this achievability result is really tight and that it holds in a more general setting: it does not need the generator matrix of the linear code to be chosen uniformly at random, it does not even need the code to be linear. To state this impossibility result, we are going to define a more general QDP problem and will overload the notation as follows
\begin{problem}[quantum decoding problem $\text{QDP}(\CC, f)$]
Take a subset $\CC$ of $\F_q^n$ and $f : \F_q^n \rightarrow \mathbb{C}$ so that $\norm{f}_2 = 1$.
We choose a codeword $\cv$ uniformly at random in $\CC$  and let  $\ket{\psi_\cv} = \sum_{\ev \in \Fqn}f(\ev)\ket{\cv+\ev}$.
	Given $(\CC, \ket{\psi_\cv})$, the goal is to recover $\cv$.
\end{problem}

This problem is more general than the previous one and even in this case if we choose the code in the best possible way we can not avoid 
the following non-achievability result.
\begin{restatable}{theorem}{intractability}\label{thm:nonachievability}
	Let $g : \F_q \rightarrow \mathbb{C}$ with $\norm{g}_2 = 1$.  Let $R$ be a fixed constant in $(0,1)$ such that $R > H_q(|\fg|^2)$. For any quantum algorithm and any code $\CC$ with $q^{Rn}$ codewords the probability $\psucc$ to solve $\QDP(\CC,g^{\otimes n})$ satisfies
	$$ \psucc \le q^{n(H_q(|\fg|^2) - R + \frac{1}{n^{1/3}})} + \negl(n) = \negl(n).$$
\end{restatable}
The strong converse theorem \cite{ON00} already shows that $\psucc = \negl(n)$, however the proof technique of Theorem \ref{thm:nonachievability} is of independent interest. The proof for the case at hand is much shorter. We will make use of the particular form of the quantum channel states $\rho_\alpha$ at hand. In terms of proof techniques, we look again at the quantum Fourier transforms of our states $\QFT{\ket{\psi_\cv}}$  and, by using the notion of typical sets, we can show that all of the states $\QFT{\ket{\psi_{\cv}}}$ approximately lie in the same subspace of dimension $q^{n H_q(|\fg|^2)}$. The condition $H_q(|\fg|^2) < R $ implies $\frac{q^{nH_q(|\fg|^2)}}{q^{Rn}} = \negl(n)$. Since there are $q^{Rn}$ different codewords, we will be able to recover any $\cv$ from $\QFT{\ket{\psi_{\cv}}}$ only with negligible probability. Formally, we perform this argument using quantum information theory techniques which are inspired by conditional quantum min-entropy notions (see~\cite{KRS09} for example).
Again, Theorem \ref{thm:nonachievability} follows from a more general result, namely Theorem  \ref{thm:nonachievability_general}, which applies to more general noise models which are non necessarily memoryless and which follows from considerations on the typical set of errors output by the channel.

\paragraph{A quantum information theoretic proof of the discrete Hirschman uncertainty principle.}
With our two theorems, we have a good understanding of the information theoretic difficulty of the quantum decoding problem. Clearly there is a huge gain in terms of the noise that can be corrected with the quantum version of the problem instead of measuring and solving the classical decoding problem: the quantum decoding problem can be solved up to a rate $R$ satisfying $R <H_q(|\fg|^2)$ whereas solving the classical problem can be achieved only up to the classical capacity of the additive channel with probability distribution $|g|^2$. The classical capacity is equal to $1-H_q(|g|^2)$ and there is in general a significant gap between the c-q capacity $H_q(|\fg|^2)$ of the channel and the classical capacity $1-H_q(|g|^2)$. In any case, the fact that the c-q capacity can not be smaller than the classical channel capacity we get after measuring the pure state  $\ket{\psi_{\cv}}$ output by the c-q channel shows the discrete Hirschman uncertainty principle \cite[Th. 23]{DCT91}
\begin{proposition}[Discrete Hirschman uncertainty principle]
  For any function $g : \F_q \mapsto \Comp$ satisfying $\pnorm{g}{2}=1$ we have
  $$
H_q(|g|^2) + H_q(|\fg|^2) \geq 1
$$
\end{proposition}
Note that this gives for the first time a (quantum) information theoretic proof of this result \cite[Th. 23]{DCT91} which was first stated and proved in a slightly more general form in \cite{DCT91} by completely other means.

\paragraph{The power of Regev's reduction with Pretty Good Measurement.}
Our fourth contribution --and this is the main contribution of this article-- is that this significant gap translates into giving much more power to Regev's reduction. As observed in \cite[\S 6 ]{CT23}, there is no hope of having a generic Regev's reduction. However in the case of Regev's reduction based on using the PGM and a Bernoulli noise, the situation is a little bit better, since with a slight tweak of the PGM, \cite{CT24} obtained an algorithm which outputs dual codewords whose weight can be as low as the minimum Hamming distance of the dual code at the information theoretic limit where the PGM works. We generalize this result here.
To state this result, we introduce two probability distributions
\begin{definition}
  We denote by $p$ the probability distribution on $\F_q^n$ equal to $(|\fg|^2)^{\otimes n}$ and by $r$ the conditional probability distribution of $p$ given that we are in $\CCps$, i.e.
  	$$r(\xv) = \left\{
\begin{array}{ll}
		0 & \mbox{if } \xv \notin  \CCps \\
\dfrac{p(\xv)}{\sum\limits_{\yv\in \CCps} p(\yv)}& \mbox{else }
\end{array}
	\right. , \text{ well defined when } \sum\limits_{\yv\in \CCps} p(\yv) \neq 0.
	$$
\end{definition}
In our general setting, a tweaked version of the PGM can be used in Regev's reduction in order to sample non zero dual codewords
according to the probability distribution $r$. This allows to sample codewords of rather large probability $p$, since we will show:\\
(i) the dual codewords $\cv$ we sample are typically such that $p(\cv) \approx q^{-n H_q\left(|\fg|^2\right)}$,\\
  (ii) we cannot do much better than this since it turns out that for most dual codes of rate $1-R$ there is no dual codeword $\cv$ of probability $p(\cv)$ greater than a quantity which is about
$q^{-nR}$.

Since the limiting rate at which the PGM works is precisely $H_q(|\fg^2|)$, we see that at this limit we obtain an algorithm outputting dual codewords $\cv$ whose probability $p(\cv)$ is close to be maximal, and this maximum is about
$q^{-nR}$. More precisely, we have
\begin{restatable}{theorem}{reductionpgm}\label{thm : reductionpgm}
  Let $g : \F_q \rightarrow \mathbb{C}$ with $\pnorm{g}{2} = 1$. Let $R$ be a real number in $(0,H_q(|\fg|^2))$. Let $f : \F_q^n \rightarrow \mathbb{C}$ be so that $f = g^{\otimes n}$.
  Let $\C$ be a linear code on $\Fqn$ with generator matrix $\Gm$ drawn uniformly at random in $\F_q^{\lfloor Rn \rfloor \times n}$. 
	For any $\epsilon > 0$:
	\begin{itemize}
	\item Algorithm \ref{algo: reductionPGM} outputs codewords  $\cv \in \CCps$  such that
          $p(\cv) \geq q^{-n(H_q(|\fg^2|)-\epsilon)}$
with probability $1-\negl(n)$.
	\item The probability over $\Gm$ that there exists a codeword $\cv$ in $\CCps$ such that $p(\cv)$ is larger than $q^{-n(R-\epsilon)}$ is a negligible function of $n$.
	\end{itemize}

\end{restatable}

There are cases where the mapping associating to $\xv \in \F_q^n$ the quantity $\log\frac{p(\zero)}{p(\xv)} =n \log{|\fg(0)|^2} - \sum_{i=1}^n \log |\fg(x_i)|^2$ is a weight $|\cdot|$ that defines a metric as
$d(\xv,\yv)= |\yv-\xv|$. The ``q-ary symmetric channel'' considered in \cite{CT24} is basically an example of this kind. There, we had  $g(0)=\sqrt{1-p}$ and $g(\alpha) = \sqrt{\frac{p}{q-1}}$ for $\alpha \neq 0$. In this case, the most likely element (according to the probability distribution $p$) is the minimum  weight codeword in $\CCps$ since $p$ is in this case obviously a decreasing function of the weight $|\cdot|$. In other words, we obtain in this case through Regev's reduction and the PGM a way to obtain a minimum weight codeword in $\CCps$ at the limiting rate where the PGM works (\ie\ at the Holevo capacity of the c-q channel). This is in strong contrast with using a classical decoder instead in Regev's reduction as shown in \cite{DRT24}. In the last case, we are really far away from getting a minimum weight codeword at the limit where classical decoding works when applied to the noise corresponding to a $q$-ary symmetric channel. In some sense, this means that in order to get a tight Regev reduction, it is really a quantum decoder (\ie\ an algorithm solving the quantum decoding problem) which is desirable rather than a classical decoder.

\paragraph{The rank metric case.} In all our proofs we make heavily use of the notion of typical set associated to the probability distribution $p \eqdef \abs{\QFT{f}}^2$. This notion applies
to more general probability distributions than those corresponding to memoryless c-q channels. We have given a more general form, namely Theorems \ref{thm:general} and \ref{thm:nonachievability_general} giving a tight achievability in the case where for the probability distribution $p$ the elements of probability $\approx q^{H_q(p)}$ carry almost all the probability of the space (meaning that they form a typical set in the usual sense). In this case again, the critical rate marking the distinction between achievability of decoding and non achievability corresponds to the entropy per symbol $\frac{H_q(p)}{n}$. We illustrate this behavior in a case which is of interest in code based cryptography which corresponds to the rank metric. This alternative in code-based cryptography to the usual Hamming metric has lead to submissions to the NIST competition for post-quantum cryptography \cite{ABDGHRTZABBBO19,mirath} which compare favorably to the corresponding submissions in Hamming metric. It can be used to obtain for instance efficient public key encryption or key exchange schemes \cite{ABDGHRTZABBBO19,ADGLRW24}, signatu re schemes \cite{mirath,BCFGJRV25},
or somewhat homomorphic encryption schemes \cite{ADG25}. In this case, the distribution of $p$ does not correspond to a memoryless channel. Nevertheless, we give a tight achievability result in this case (see Theorems \ref{thm:achievability_rm} and \ref{thm : inachievability_rm}) as a corollary of Theorems \ref{thm:general} and \ref{thm:nonachievability_general}. The probability distribution $p$ used in this case is a decreasing function of the rank weight. We also give a more general form of Theorem \ref{thm : reductionpgm} showing that Regev's reduction when applied to a variant of the PGM yields elements of the dual code of maximal probability
at the limit where the PGM works. This more general result applies also to the rank metric case and shows that again we get at the limit where the PGM works elements in the dual code of
minimum rank weight.

\paragraph{The decoding problem in the surjective regime and an operational version of the duality result of J. Renes \cite{R18a}}
If we replace in Regev's reduction procedure the initial state $ \frac{1}{\sqrt{|\C|}} \sum_{\cv \in \C} \sum_{\ev \in \F_q^n}  f(\ev) \ket{\cv}\ket{\cv+\ev}$ by $\frac{1}{\sqrt{|\C|}} \sum_{\cv \in \C, \ev \in \F_q^n} \chi_\cv(-\vv)f(\ev) \ket{\cv}\ket{\cv + \ev}$, solve the quantum decoding problem with the PGM as in Algorithm 4 (for more details see \S \ref{ss:Renes}), then we 
output elements $\vv + \CCp$  of maximal probability $p$ at the limit where PGM works. More precisely we obtain
\begin{restatable}{theorem}{reductionpgmwithv}\label{thm : reductionpgmwithv}
  Let $g : \F_q \rightarrow \mathbb{C}$ with $\pnorm{g}{2} = 1$. Let $R$ be a real number in $(0,H_q(|\fg|^2))$. Let $f : \F_q^n \rightarrow \mathbb{C}$ be so that $f = g^{\otimes n}$.
  Let $\vv$ be an element in $\F_q^n$.
  Let $\C$ be a linear code on $\Fqn$ with generator matrix $\Gm$ drawn uniformly at random in $\F_q^{\lfloor Rn \rfloor \times n}$. 
	For any $\epsilon > 0$:
	\begin{itemize}
	\item Algorithm 4.4 outputs elements  $\yv \in \vv+ \CCp$  such that
          $p(\yv) \geq q^{-n(H_q(|\fg^2|)-\epsilon)}$
with probability $1-\negl(n)$.
	\item The probability over $\Gm$ that there exists an element $\yv$ in $\vv+ \CCp$ such that $p(\yv)$ is larger than $q^{-n(R-\epsilon)}$ is a negligible function of $n$.
	\end{itemize}
\end{restatable}
This gives an optimal lossy compression procedure for the dual code $\CCp$ based on a optimal decoding procedure based on PGM for $\C$ for the classical-quantum c-q channel corresponding to $g$. In some sense it can be viewed (with the exception of the
$\sqrt{2 \epsilon}$ error term which appears in \cite[\S 5, Cor. 2]{R18a}) as an explicit version of the lossy compression procedure mentioned in \cite[\S 5, Cor. 2]{R18a}.

	\section{Notation and Preliminaries}
	\subsection{Notation}
\label{ss:notation}
\paragraph{General notation.}
The finite field with $q$ elements is denoted by $\Fq$. The set of nonnegative integers is denoted by $\N$.
The cardinality of a finite set $\Ec$ is denoted by $|\Ec|$.
For a vector space $W$ and subset $V$ of it, we write that $V \leq W$ iff $V$ is a subspace of $W$.
A {\em negligible function} is a function $\eta:\N \rightarrow \Real$ such that for every positive integer $a$ there exists an integer $N$ such that for all $n \geq N$ we have
$\eta(n) \leq \frac{1}{n^a}$. When for a function $f:\N \rightarrow \Real$ we write that $f(n) = \negl(n)$, this means that $f$ is a negligible function.

\paragraph{Vector and matrices.}
For a Hermitian matrix $\Mm$, we write that $\Mm \succcurlyeq 0$ when $\Mm$ is positive semidefinite. Vectors are row vectors as is standard in the coding community and $\trsp{\xv}$ (resp. $\trsp{\Mm}$) denotes the transpose of a vector (resp. of a matrix). In particular, vectors will always be denoted by bold small letters and matrices with bold capital letters. The inner product
between two vectors $\xv=(x_i)_{1 \leq i \leq n }$ and $\yv=(y_i)_{1 \leq i \leq n }$ of $\F_q^n$ is denoted by $\xv \cdot \yv \eqdef \sum_{i=1}^n x_i y_i$. $\Im$ denotes the identity matrix. Its size will be clear from the context.

\paragraph{The classical and quantum Fourier transform on $\Fqn$.}
In this article, we will use the quantum Fourier transform on $\Fqn$.
It is based on the characters of the group $(\Fqn,+)$ which are defined as follows (for more details see \cite[Chap 5, \S 1]{LN97}, in particular a description of the characters in terms of the trace function is given in  \cite[Ch. 5, \S 1, Th. 5.7]{LN97}).
\begin{definition}
	Fix $q = p^s$ for a prime integer $p$ and an integer $s \ge 1$. The characters of $\F_q$ are the functions $\carf{y} : \F_q \rightarrow \mathbb{C}$ indexed by elements $y \in \F_q$ defined as follows
	\begin{eqnarray*}
		\car{y}{x} & \eqdef & e^{\frac{2i \pi \tr(x  y)}{p}}, \quad \text{with} \\
		\tr(a) & \eqdef & a + a^p + a^{p^2} + \dots + a^{p^{s-1}}.
	\end{eqnarray*}
	 We extend the definition to vectors $\xv,\yv \in \F_q^n$ as follows:
	$$ \car{\yv}{\xv} \eqdef \prod\limits_{i = 1}^n \car{y_i}{x_i}= e^{\frac{2i \pi \tr(\xv \cdot \yv)}{p}}.$$
\end{definition}
When $q$ is prime, we have $\car{y}{x} = e^{\frac{2i\pi xy}{q}}=\omega_q^{xy}$, where $\omega_q \eqdef e^{\frac{2i\pi}{q}}$.

\begin{definition}
	Let $f : \F_q^n \rightarrow \mathbb{C}$. We define $\hf : \F_q^n \rightarrow \mathbb{C}$ such that 
	$\hf(\yv) = \frac{1}{\sqrt{q^n}}\sum_{\xv \in \F_q^n} \chi_{\yv}(\xv) f(\xv)$.
\end{definition}
The way we normalize the Fourier transform $\hf$ of a function $f$ ensures that $\norm{f}_2 = \norm{\hf}_2$. The quantum Fourier transform over $\F_q^n$ is the quantum unitary such that for each $f : \F_q^n \rightarrow \mathbb{C}$ such that $\norm{f}_2 = 1$, we have 
$$ \QFTt_{\F_q^n} \left(\sum_{\xv \in \F_q^n} f(\xv)\ket{\xv}\right) = \sum_{\xv \in \F_q^n} \hf(\xv)\ket{\xv}.$$

\paragraph{Linear codes.}
A linear code over $\Fq$ of length $n$ is a subspace of $\F_q^n$. We say that it is an $[n,k]$ code when its dimension is $k$. It can be specified by a generating matrix $\Gm \in \F_q^{k \times n}$, in which case $\C = \lbrace \sv \Gm : \sv \in \Fqk \rbrace$. We also define $\C^* = \C \backslash \{\zero\}$.
The dual code $\C^\perp$ which is defined as
$\C^\perp = \{\yv \in \F_q^n: \yv \cdot \cv = 0\}$ can also be described as
\begin{equation}
  \label{eq:Cperp_car}
  \C^\perp = \{\yv \in \F_q^n: \car{\yv}{\cv}=1,\;\forall \cv \in \C\}.
\end{equation}
When the code length $n$ is a product $n = a \cdot b$ we can arrange the entries of a vector $\xv=(x_i)_{0 \leq i < n} \in \F_q^n$ in a matrix $\Xm = \matv(\xv) \in \F_q^{a \times b}$. This is obtained by defining the entry $X_{i,j}$
in the $i$-th row and the $j$-th column of $\Xm$ as $X_{i,j} = x_{i\cdot b +j}$ where $i \in \{0,\cdots,a-1\}$, $j \in \{0,\cdots,b-1\}$. 
This allows us to define the rank weight on $\F_q^{n}$ as
$$|\xv|_\rk \eqdef \rank(\matv(\xv)).$$

\paragraph{}

\subsection{Information Theory}
We first recall the notion of entropy
\begin{definition}\label{def:q-entropy}
  Let $p$ be a probability distribution on a finite alphabet $\Ac$.
Its $q$-ary entropy is defined as
	$$H_q(p) \eqdef -\sum_{x \in \Ac}p(x)\log_qp(x).$$
\end{definition}
We use this notion to define the typical set of a probability distribution over $\F_q^n$:
\begin{definition}[typical set]\label{def:typical}
  The typical set $\Typ(p)$ corresponding to the probability distribution $p$ over $\F_q^n$, gap $\epsilon$ where $\epsilon$ is a positive number and bounding functions
  $A(\epsilon)$ and $B(\epsilon)$ is defined by
  $$\Typ(p) = \lbrace \yv \in \Fqn : A(\epsilon) \leq p(\yv) \leq  B(\epsilon)\rbrace.$$
  We will generally drop the dependence in $p$ which will be clear from the context and simply write 
  $\Typ$. The defect  $\delta(\epsilon,n)$ is the function of $\epsilon$ defined by $1-p(\Typ)$.
\end{definition}  

A direct consequence of this definition is that we have
\begin{lemma}\label{lem:cardinality}
  For all $\epsilon > 0$ and positive integer $n$ we have
  $$
(1 - \delta(\epsilon,n)) \frac{1}{B(\epsilon)} \leq |\Typ| \leq \frac{1}{A(\epsilon)}.
  $$
\end{lemma}
The rationale behind this definition is that for many families of probability distributions over $\F_q^n$ we have $p(\Typ)=1-o(1)$ for any fixed $\epsilon$ as $n$ tends to infinity
for bounding functions $A(\epsilon)$ and $B(\epsilon)$ which are close to each other.
 In particular it holds for product probability distributions $p = r^{\otimes n}$ for fixed $r$ as is well known
\COMMENT{\textbf{Shannon's entropy} of a random variable quantifies the uncertainty of the possible outcomes of the variable. Given a random variable $X$ which takes values in an alphabet $\Ac$ distributed according to $p : \Ac \rightarrow \lc0,1\rc$ in an alphabet $\Ac$ and $ q = \abs{\Ac}$, the normalized entropy is

	$$H_q(X) = -\sum_ {x \in \Ac}p(x)\log_qp(x)$$}

\begin{proposition}\cite[Section 3.1]{CT06}\label{Proposition:1}

	Let $r : \F_q \rightarrow [0,1]$ be a probability function on $\F_q$. Let $p = r^{\otimes n}$ and $\Typ$ be the associated $\epsilon$-typical set for some $n \in \mathbb{N}$ and $\epsilon > 0$ corresponding to the bounding functions $A(\epsilon)=q^{-n(H_q(r)+\epsilon)}$ and $B(\epsilon) = q^{-n(H_q(r)-\epsilon)}$.

	\begin{enumerate} 
		\item $\delta(\eps,n) \le e^{\frac{-2n\eps^2}{K}}$ for some constant $K$ that depends only on $r$. In particular, when $\eps = \frac{1}{n^{1/3}}$, then, $\delta(\eps,n) = \negl(n)$.
		\item $(1 - \delta(\eps,n))q^{n(H_q(r)-\epsilon)} \leq \abs{\Typ}\leq q^{n(H_q(r)+\epsilon)}$.
	\end{enumerate}
\end{proposition}

\begin{proof}
  Fix $\Typ$ for some $r : \F_q \rightarrow [0,1], n \in \mathbb{N}$ and $\eps > 0$. Let $p = r^{\otimes n}$. Let $\xv=(x_1,\cdots,x_n)$ be a random variable on $\F_q^n$ distributed according to $p$. We have
	\begin{align*}
		\delta(\eps,n) & = \Pr\left[\xv \notin\Typ\right] \\
		& = \Pr \left[ \left|\sum\limits_{i=1}^n - \log_q(r(x_i)) -  nH_q(r) \right| \geq n\epsilon\right]
	\end{align*}
     
	Let $X_1,\dots,X_n$ be random i.i.d. variables such that 
	$$\Pr[X_i = \alpha] = 
	\left\{
	\begin{array}{cl}
		-\log_q(r(\alpha)) & \textrm{ if } r(\alpha) \neq 0 \\
		0 & \textrm{ if } r(\alpha) = 0
	\end{array} \right.$$

	Let $S_n = \sum_{i = 1}^n X_i$. Notice that $\esp[X_i] = H_q(r)$ hence $\esp[S_n] = nH_q(r)$. Also, notice that $0 \le X_i \le K$ where $K = \max\limits_{\alpha : r(\alpha) \neq 0}\{- \log_q(r(\alpha))\}$. Using Hoeffding's inequality\footnote{For any  i.i.d random variables $X_1, ... X_n$ such that $\Pr(a_i\leq X_i \leq b_i)=1$, for any $t>0$, 
	$
		\Pr(S_n - \esp[S_n]\geq t)\leq \exp\lp-\frac{2t^2}{M}\rp$
	with $M = \sum\limits_{i=1}^n(b_i - a_i)^2$ and $S_n = \sum\limits_{i=1}^n X_i$.}, we write
	\begin{align*}
		\delta(\epsilon,n) & = \Pr(S_n - \esp[S_n]\geq t)\leq e^{-\frac{2t^2}{n K^2}} = e^{-\frac{2\eps^2 n}{K^2}}
	\end{align*}
	The second point follows directly from Lemma \ref{lem:cardinality}.
\end{proof}

We will also need the following results for the variance and the expectation of the cardinality of the intersection of a random linear code on $\F_q^{n}$ with an arbitrary subset of $\Fqn$.

\begin{lemma}\label{Variance}(Intersection Expectation and Variance lemma ) \cite[Lemma 1.1]{B97b}
	Let $\C$ be a random linear code on $\Fqn$ whose parity-check matrix $\Hm$ is chosen uniformly at random in $\F_q^{(n-k)\times n}$.   Let $E$ be an arbitrary subset of $\Fqn$. Then
	$$\frac{\abs{E}}{q^{n-k}} \leq \mathbb{E}_{\Hm}(\vert \C \cap E \vert) \leq \frac{\abs{E} - 1}{q^{n-k}} + 1$$
	$$\Var_\Hm(\vert \C \cap E \vert)  \leq (q-1) \mathbb{E}_{\Hm}(\vert \C \cap E \vert)$$
		
\end{lemma}

\subsection{Quantum information theory}\label{ss:quantum_information}
\begin{definition}[von Neumann entropy] The von Neumann entropy $S(\rho)$ of a quantum state $\rho$ living in a $d$-dimensional Hilbert space is defined as $S(\rho) \eqdef - \sum_{i=1}^d \lambda_i \log_d \lambda_i$ where
  $\rho = \sum_{i=1}^d \lambda_i \ketbra{e_i}{e_i}$ is the spectral decomposition of $\rho$. 
  \end{definition}
\begin{definition}
	For two quantum states $\rho,\sigma$, we define their trace distance  $D(\rho,\sigma) \eqdef \frac{1}{2} \tr(|\rho - \sigma|)$ where $|A|$ is the positive square root of $A^\dagger A$. 
\end{definition}
\begin{claim}[\cite{NC10}]
	If $\rho = \kb{\phi}$ and $\sigma = \kb{\psi}$, then $D(\rho,\sigma) = \sqrt{1 - |\braket{\phi}{\psi}|^2}$.
\end{claim}
\begin{claim}[\cite{NC10}]\label{Claim:TD}
	$D(\rho,\sigma) = \max_{P} \{\tr(P \rho - P \sigma)\},$ where the maximum is over all positive operators $P \preccurlyeq I$.
\end{claim}
\COMMENT{\begin{definition}
	For two quantum states $\rho,\sigma$, we define their fidelity  $F(\rho,\sigma) \eqdef \tr(\sqrt{\rho^{1/2}\sigma\rho^{1/2}}).$
\end{definition}
\begin{claim}
	For two states $\rho = \frac{1}{|\C|} \sum_{\cv \in \C} \kb{\cv} \otimes  \ket{\psi_\cv}\bra{\psi_\cv}$ and $\sigma = \frac{1}{|\C|} \sum_{\cv \in \C} \kb{\cv} \otimes  \ket{\phi_\cv}\bra{\phi_\cv}$, we have 
	$F(\rho,\sigma) =  \frac{1}{|\C|} \sum_{\cv \in \C} |\braket{\psi_{\cv}}{\phi_{\cv}}|.$
\end{claim}
\begin{proof}
	First, we can write $\rho^{1/2} = \frac{1}{\sqrt{|\C|}} \sum_{\cv \in \C} \kb{\cv} \otimes  \ket{\psi_\cv}\bra{\psi_\cv}$. From there, we compute
	\begin{align*}
		F(\rho,\sigma) & = \tr(\sqrt{\rho^{1/2}\sigma\rho^{1/2}}) \\*
		& = \frac{1}{{|\C|}}\tr\left(\sqrt{\sum_{\cv \in \C} \kb{\cv} \otimes \left(\kb{\psi_\cv} \cdot \kb{\phi_{\cv}} \cdot \kb{\psi_{\cv}}\right)}\right) \\
		& = \frac{1}{{|\C|}}\sum_{\cv \in \C} \tr\left(\sqrt{ \kb{\psi_\cv} \cdot \kb{\phi_{\cv}} \cdot \kb{\psi_{\cv}}}\right) \\
		& = \frac{1}{{|\C|}}\sum_{\cv \in \C} \tr\left(\sqrt{|\braket{\psi_\cv}{\phi_{\cv}}|^2 \kb{\psi_\cv}}\right) \\
		& = \frac{1}{{|\C|}}\sum_{\cv \in \C} |\braket{\psi_\cv}{\phi_{\cv}}|
	\end{align*}
\end{proof}
\begin{claim}[Fuchs-van de Graaf inequalities~\cite{NC10}]
	For any quantum states $\rho,\sigma$, we have
	$$ 1 - F(\rho,\sigma) \le D(\rho,\sigma) \le \sqrt{1 - F^2(\rho,\sigma)}.$$
\end{claim}}
\begin{claim}\label{Claim:Trace}
  Let $A,B$ be two positive semidefinite matrices. We have $\tr(AB) \ge 0$. As a direct corollary, for all $A$, $B$ and $C$ hermitian matrices such that
  $A \succcurlyeq 0$ and $B \preccurlyeq C$, we have that $\tr(AB) \le \tr(AC)$.
\end{claim}
\begin{proof}
	We write $\tr(AB) = \tr(B^{1/2}AB^{1/2}) \ge 0$, where the inequality comes from the fact that $B^{1/2}AB^{1/2}$ is also positive semidefinite.
\end{proof}
\subsection{The Pretty Good Measurement} 
\begin{definition} Let $J$ be a set of quantum pure states $\lbrace \ket{\psi_j}\rbrace_{1 \leq j \leq |J|}$. Let $\rho = \sum_{j = 1}^{|J|}\ket{\psi_j}\bra{\psi_j}$. The Pretty Good Measurement associated to this set  is the POVM $\lbrace M_j\rbrace_{1 \leq j \leq |J|}$ with $$M_j = \rho^{-1/2}\ket{\psi_j}\bra{\psi_j}\rho^{-1/2}$$

	From now on, we will refer to the Pretty Good Measurement as PGM.

\end{definition}
One can check that we have indeed $\sum_{j} \rho^{-1/2} \kb{\psi_j} \rho^{-1/2} = \Im$. In our case, it turns out that the PGM is actually optimal for attaining the highest average probability of success.
\begin{proposition}
  Assume that the $\ket{\psi_j}$ are equally probable, independent and form a geometrically uniform set of states~\cite[\S VIII, A.]{EF01}, in the sense that there exists an Abelian group  $G$ of unitaries such that
  $\{\ket{\psi_j},\; j \in J\} = \{U \ket{\psi_{j_0}},\; U \in G\}$ for some $j_0 \in J$. Then the PGM applied to the $\ket{\psi_j}$'s has the largest average probability of success.
\end{proposition}
This was proved in \cite[prop. 1]{BKMH97} in the case of an Abelian group generated by a single element and the more general form stated here is given in~\cite[\S VIII, C, Th. 4]{EF01}.
It is also a consequence of~\cite[Th. 1]{SKIH98}.

\COMMENT{
	We define the quantum decoding problem as in \cite{CT23}
	\begin{definition}
		For $q,n,k \in \N^*$, with $q \geq 2$, for $f : \Fqn \rightarrow \Comp $ with $\norm{f}_2 = 1$, we define \textbf{the distribution }$\D_\mathcal{Q}(q,n,k,f)$ sampled as follows : $G \Unif \lbrace 0,1 \rbrace^{k\times n}$, $m \Unif \Fqk$, $c = mG$, $\ket{\psi_c} = \sum_{e \in \Fqn}f(e)\ket{c+e}$, return $(G,\ket{\psi_c}, c)$.
	\end{definition}
	
	\begin{definition}
		For $q,n,k \in \N^*$, with $q \geq 2$, for $f : \Fqn \rightarrow \Comp $ with $\norm{f}_2 = 1$, the decoding problem QDP$(q,n,k,f)$ is the following. We sample $(G,\ket{\psi_c}, c) \Unif \D_\mathcal{Q}(q,n,k,f)$ and the goal is, given only $(G, \ket{\psi_c})$, to recover $c$.
	\end{definition}}

 \section{(Non-)Achievability Results}

Our goal is here to first give a proof of the achievability result, namely Theorem \ref{thm:achievability}.  This will be obtained by computing the probability of success of the Pretty Good Measurement when applied to a random linear code. It will show that as soon as the rate of the linear code is below $H_q(|\fg|)^2$, then we can solve the quantum decoding problem with high probability. On the other hand, the non-achievability result, namely Theorem \ref{thm:nonachievability}, will apply to {\em any} code and shows that as soon as the rate of the code is above this limit, we can not hope to solve the quantum decoding problem with non-negligible probability anymore.
Finally, we will apply these two results in the case of rank metric where we get similar results as for the Hamming metric \cite{CT24}, where the limit of achievability corresponds to
sample dual codewords of minimum rank distance.

\subsection{Notation and Preliminaries}\label{ss:preliminaries}

It turns out that for both results it will be very helpful to apply the quantum Fourier transform to the states $\ket{\psi_\cv}$ that we want to distinguish in the quantum decoding problem. In the whole section, $\CC$ denotes the code that we want to decode and $f$ denotes the  associated noise distribution over $\F_q^n$. It  is such that $\norm{f}_2=1$.
The input states to the decoding problem are denoted by $\ket{\psi_\cv}$ for $\cv \in \CC$, \ie 
$$
\ket{\psi_\cv} \eqdef \sum_{\ev \in \Fqn}f(\ev)\ket{\cv+\ev}.
$$
Note that these states are just shifts of the ``noise state'' $\sum_{\ev \in \Fqn}f(\ev)\ket{\ev}$ which explains why the quantum Fourier transform applied to the $\ket{\psi_\cv}$'s have all the same amplitudes up to a phase term as shown by 
\begin{lemma} \label{lem:qft_c}
We have 
\begin{equation}\label{eq:qft_c}
 \QFT{\ket{\psi_{\cv}}} = \sum_{\yv \in \F_q^n} \hf(\yv) \chi_{\cv}(\yv) \ket{\yv}.
 \end{equation} 
\end{lemma}

\begin{proof}
\eqref{eq:qft_c} follows on the spot from observing that 
\begin{eqnarray*}
	\QFT{\ket{\psi_{\cv}}} & = & \QFTt_{\F_q^n} \left(\sum_{\ev \in \F_q^n}  f(\ev) \ket{\cv + \ev} \right)\\
& = & \frac{1}{\sqrt{q^n}}  \sum_{\yv \in \F_q^n}   \left(\sum_{\ev \in \F_q^n} \chi_{\cv + \ev}(\yv) f(\ev) \right) \ket{\yv} \\
& = & \sum_{\yv \in \F_q^n}   \underbrace{\frac{1}{\sqrt{q^n}}\sum_{\ev \in \F_q^n} \chi_{\ev}(\yv) f(\ev)}_{\hf(\yv)} \chi_{\cv}(\yv) \ket{\yv} \\
& = &   \sum_{\yv \in \F_q^n} \hf(\yv)   \chi_{\cv}(\yv) \ket{\yv}.
\end{eqnarray*}
\end{proof}

\paragraph{The capacity of the classical-quantum pure state channel in the memoryless case.} The Fourier transform turns out to be really helpful for computing the capacity of the c-q channel capacity.
Indeed we have \begin{proposition}\label{pro:cq-capacity}
We assume that $f= g^{\otimes n}$ where $g$ is a function of $L_2$ norm 1 defined over $\F_q$.
The Holevo capacity  of the c-q channel mapping any $\cv \in \F_q^n$ to the pure state $\ket{\psi_\cv}$ is equal to
$H_q\left(|\fg|^2\right)$.
\end{proposition}
\begin{proof}
The Holevo capacity $C_{\Hc}$ of this pure state channel is equal to $\sup_{p_\alpha} S\left(\sum_{\alpha \in \F_q} p_\alpha \rho_\alpha \right)$ where the supremum is taken over all
probability distributions $(p_\alpha)_{\alpha \in \F_q}$ over $\F_q$ and
\begin{eqnarray*}
\rho_\alpha &\eqdef &\ketbra{\psi_\alpha}{\psi_\alpha}\;\;\text{with}\\
\ket{\psi_\alpha} &\eqdef &\sum_{u \in \F_q} g(u) \ket{\alpha + u}.
\end{eqnarray*}
For $\beta \in \Fq$ we denote by $U_\beta$ the unitary mapping $\ket{\alpha}$ to $\ket{\alpha+\beta}$  for any $\alpha \in \F_q$. Let us denote by $(p_\alpha)_{\alpha \in \Fq}$ the probability
distribution that attains the supremum of  $S\left(\sum_{\alpha \in \F_q} p_\alpha \rho_\alpha \right)$. If there are several distributions that attain the supremum, we choose arbitrarily one of them. Let us observe that for any $\beta \in \Fq$
\begin{eqnarray}
C_{\Hc}=   S\left(\sum_{\alpha \in \F_q} p_\alpha \rho_\alpha \right) & = &   S\left(U_\beta \left(\sum_{\alpha \in \F_q} p_\alpha \rho_\alpha\right) U_\beta^\dagger \right) \nonumber\\
& = &   S\left(\sum_{\alpha \in \F_q} p_\alpha U_\beta\ketbra{\psi_\alpha}{\psi_\alpha}U_\beta^\dagger \right) \nonumber \\
& = &   S\left(\sum_{\alpha \in \F_q} p_\alpha \rho_{\alpha+\beta} \right)\nonumber\\
& = &   S\left(\sum_{\alpha \in \F_q} q^\beta_\alpha \rho_\alpha \right) \label{eq:group}
\end{eqnarray}
where $q^\beta$ is the probability distribution over $\F_q$ defined by $q^\beta_\alpha = p_{\alpha-\beta}$. Let $\sigma_\beta  \eqdef \sum_{\alpha \in \F_q} q^\beta_\alpha \rho_\alpha$. The above can be therefore rewritten $\forall \beta \in \F_q, \ C_{\Hc} = S(\sigma_\beta)$.
By using now the concavity of the von Neumann entropy we deduce that
\begin{eqnarray*}
  C_{\Hc} = \frac{1}{q} \sum_{\beta \in \F_q} S(\sigma_\beta) \le  S\left(\frac{1}{q}\sum_{\beta \in \Fq} \sigma_\beta\right) =  S\left(\frac{1}{q}\sum_{\beta \in \Fq} \sum_{\alpha \in \F_q} q^\beta_\alpha \rho_\alpha\right) = S\left( \frac{1}{q} \sum_{\alpha\in \Fq} \rho_\alpha\right).
\end{eqnarray*}
We deduce that the Holevo capacity is equal to $  S(\rho)$ where $\rho \eqdef \frac{1}{q} \sum_{\alpha\in \Fq} \rho_\alpha$. It is convenient here to express these density matrices in the Fourier
basis, since $\rho$ is diagonal in this basis:
\begin{eqnarray}
S(\rho) & = & S\left( \frac{1}{q} \sum_{\alpha \in \Fq} \ketbra{\psi_\alpha}{\psi_\alpha}\right) \nonumber\\
& = & S\left( \frac{1}{q} \sum_{\alpha \in \Fq} \ketbra{\QFT{\psi_\alpha}}{\QFT{\psi_\alpha}}\right) \nonumber\\
& = & S\left( \frac{1}{q} \sum_{\alpha \in \Fq} \sum_{\beta \in \Fq} \sum_{\gamma \in \Fq} \chi_\alpha(\beta)\fg(\beta)\overline{\chi_\alpha(\gamma)\fg(\gamma)}\ketbra{\beta}{\gamma}\right)\;\;\text{ (by using Lemma \ref{lem:qft_c})} \nonumber\\
& = & S\left( \frac{1}{q} \sum_{\alpha \in \Fq} \sum_{\beta \in \Fq} \sum_{\gamma \in \Fq} \chi_\alpha(\beta-\gamma)\fg(\beta)\overline{\fg(\gamma)}\ketbra{\beta}{\gamma}\right)  \nonumber \\
& = & S\left(\sum_{\beta \in \Fq} |\fg(\beta)|^2 \ketbra{\beta}{\beta}\right)\label{eq:final}\\
& = & H_q(|\fg|^2).\nonumber
\end{eqnarray}
\eqref{eq:final} is consequence of the fact
\begin{equation*}
\sum_{\alpha \in \Fq} \chi_\alpha(\beta- \gamma)  =  \left\{\begin{array}{lcl} 0&\;\,&\text{if $\beta \neq \gamma$}\\q & \;\; & \text{otherwise.}\end{array}\right.
\end{equation*}
\end{proof}

\paragraph{Probability of success of the Pretty Good Measurement.}
Applying a unitary transform does not change the probability of success of pretty good measurement and therefore the probability of success of the pretty good measurement $\PPGM$  for recovering $\cv$ from $\ket{\psi_\cv}$ is equal to the probability of success of the pretty good measurement for recovering $\cv$ from $\QFT{\ket{{\psi}_\cv}}$.
Let us first rewrite a little bit the expression of $\QFT{\ket{\psi_\cv}}$ by using the fact that $\chi_\cv(\ev)$ is constant on the cosets of $\C^\perp$ for any $\cv$ in $\C$ since
$\chi_\cv(\cv^\perp)=1$ for any $\cv$ in $\C$ and $\cv^\perp$ in $\C^\perp$. For any $\sv \in \F_q^n/C^\perp$ we fix an arbitrary element $\uv_\sv$ in $\sv+\C^\perp$. Therefore
\begin{align*}
	\QFT{\ket{\psi_\cv}} &=  \sum_{\ev \in \Fqn}\QFT{f}(\ev) \chi_{\cv}(\ev)\ket{\ev}\\
	&= \sum_{\sv \in \Fqn/\C^\perp}\sum_{\ev \in\sv + \C^\perp }\QFT{f}(\ev) \chi_{\cv}(\ev)\ket{\ev}\\
	&= \sum_{\sv \in \Fqn/\C^\perp}\chi_{\cv}(u_\sv) \sum_{\ev \in\sv + \C^\perp}\QFT{f}(\ev)\ket{\ev}\\
	&= \sum_{\sv \in \Fqn/\C^\perp}\chi_{\cv}(u_\sv) \ket{W_\sv},\\
        \text{where }\ket{W_\sv} & \eqdef \sum_{\ev \in\sv + \C^\perp}\QFT{f}(\ev)\ket{\ev}.
\end{align*}

\begin{proposition}\label{prop:PPGM}
The PGM succeeds to recover $\cv$ from $\ket{\QFT{\psi_{\cv}}}$ 
	with probability $\PPGM = \frac{1}{q^k}\lp\sum\limits_{\sv \in \Fqn/\C^\perp }Z_\sv\rp^2$ where $Z_\sv\eqdef \norm{\ket{W_\sv}}$
        
\end{proposition}

\begin{proof}
To get the probability of success of the PGM, we need to define the operators of the measurement. For that, we will use the same technique as in \cite[\S 5.1]{CT23}. 
	The operator of the PGM associated to the state $\ket{\QFT{\psi_\cv}}$ is defined as follows:
	$$M_\cv = \rho^{-1/2}\ket{\QFT{\psi_\cv}}\bra{\QFT{\psi_\cv}}\rho^{-1/2} \text{ with } \rho = \sum_{\cv \in \C}\ket{\QFT{\psi_\cv}}\bra{\QFT{\psi_\cv}}$$
	We write $$ \rho = \sum_{\cv \in \C}\ket{\QFT{\psi_\cv}}\bra{\QFT{\psi_\cv}} = q^k \sum_{\sv \in \Fqn/\C^\perp} Z_\sv^2\ket{\widetilde{W}_\sv}\bra{\widetilde{W}_\sv}$$
        where $\ket{\widetilde{W}_\sv} = \dfrac{1}{Z_\sv}\ket{W_\sv}$.
	Then, $\rho^{-1/2} = \dfrac{1}{\sqrt{q^k}}\sum_{\sv \in \Fqn/\C^\perp}\dfrac{1}{Z_\sv}\ket{\widetilde{W}_\sv}\bra{\widetilde{W}_\sv}$.\\
	We get \begin{align*}
		\rho^{-1/2}\ket{\QFT{\psi_c}} &= \frac{1}{\sqrt{q^k}}\sum_{\sv, \sv' \in \Fqn/\C^\perp}\frac{Z_{\sv'}}{Z_\sv}\chi_\cv(u_\sv)\ket{\widetilde{W}_\sv}\braket{\widetilde{W}_\sv}{\widetilde{W}_{\sv'}}\\
		&= \frac{1}{\sqrt{q^k}}\sum_{\sv \in \Fqn/\C^\perp}\chi_\cv(u_\sv)\ket{\widetilde{W}_\sv} 
	\end{align*}
	using the fact that the $\ket{\widetilde{W}_\sv}$ are pairwise orthogonal.\\
	Then, we can write \begin{align*}
	 \bra{\QFT{\psi_\cv}} \rho^{-1/2}\ket{\QFT{\psi_\cv}} &= \frac{1}{\sqrt{q^k}}\sum_{\sv, \sv' \in \Fqn/\C^\perp}\chi_\cv(u_{\sv'}-u_\sv)\braket{\widetilde{W}_\sv}{\widetilde{W}_{\sv'}}\\
	 &=  \frac{1}{\sqrt{q^k}}\sum_{\sv \in \Fqn/\C^\perp}\braket{\widetilde{W}_\sv}{W_{\sv}}
	\end{align*}
	Finally, we get the probability of success of the PGM:
	\begin{align*}
	p_\cv &= \bra{\QFT{\psi_\cv}}M_\cv \ket{\QFT{\psi_\cv}} \\
	&=  \bra{\QFT{\psi_\cv}} \rho^{-1/2}\ket{\QFT{\psi_\cv}}\bra{\QFT{\psi_\cv}}\rho^{-1/2} \ket{\QFT{\psi_\cv}}\\
	&= \frac{1}{q^k}\left|\sum_{\sv \in \Fqn/\C^\perp } \braket{\widetilde{W}_\sv}{W_\sv}\right|^2 = \frac{1}{q^k}\lp\sum_{\sv \in \Fqn/\C^\perp }Z_\sv\rp^2
	\end{align*}
\end{proof}

\paragraph{Nice family of typical sets.}
All the results of this article rely heavily on the notion of typical sets. The following definition will turn out to be very useful in the rest of the paper.

\begin{definition}[nice family of typical sets] Let $(n_i)_{i \geq 1}$ be an increasing set of positive integers and let $(p_i)_{i \geq 1}$ be a family of probability distributions where for every positive integer $i$, $p_i$ is a probability distribution over $\F_q^{n_i}$. We say that the sequence $(p_i)_{i \geq 1}$ admits a nice family of typical sets $\left\{\Typc{n_i}(p_i), i \in \N^*, \epsilon >0\right\}$ if there exist a constant $H$, two functions $\alpha$ and $\beta$ defined over $\Real^+$ and ranging over $\Real^+$ whose limit at $0$ is equal to $0$ and two positive constants $K_1$ and $K_2$ such that the defect function $\delta$ and the bounding functions $A_i(\epsilon)$ and $B_i(\epsilon)$ of the typical set $\Typc{n_i}$
satisfy
\begin{enumerate}
\item $\delta(\epsilon,n_i) = \negl(i)$ for every $\epsilon >0$
\item $A_i(\epsilon) \geq K_1 q^{-n_i(H+\alpha(\epsilon))}$ for every positive integer $i$ and any $\epsilon >0$,
\item $B_i(\epsilon) \leq K_2 q^{-n_i(H-\beta(\epsilon))}$ for every positive integer $i$ and any $\epsilon >0$.
\end{enumerate}
\end{definition}

For the proof of the strong converse of the achievability result, it will be helpful to allow the $\epsilon$ involved in the definition to go to $0$ as $i$ tends to infinity and
still require that the defect function is a negligible function of $i$. This gives the following definition
\begin{definition}[nice family of typical sets in the strong sense] Let $(n_i)_{i \geq 1}$ be an increasing set of positive integers and let $(p_i)_{i \geq 1}$ be a family of probability distributions where for every positive integer $i$, $p_i$ is a probability distribution over $\F_q^{n_i}$. We say that the sequence $(p_i)_{i \geq 1}$ admits a nice family of typical sets in the strong sense
$\left\{\Typii(p_i), i \in \N^*\right\}$ where $\epsilon$ is a function from $\N^*$ to $\Real^+$ where $\epsilon(i)$ tends to $0$ as $i$ tends to infinity if there exist a constant $H$, two functions $\alpha$ and $\beta$ defined over $\Real^+$ and ranging over $\Real^+$ whose limit at $0$ is equal to $0$ and two positive constants $K_1$ and $K_2$ such that the defect function $\delta$ and the bounding functions $A_i(\epsilon)$ and $B_i(\epsilon)$ of the typical set $\Typc{n_i}$
satisfy
\begin{enumerate}
\item $\delta(\epsilon(i),n_i) = \negl(i)$,
\item $A_i(\epsilon) \geq K_1 q^{-n_i(H+\alpha(\epsilon))}$ for every positive integer $i$ and any $\epsilon >0$,
\item $B_i(\epsilon) \leq K_2 q^{-n_i(H-\beta(\epsilon))}$ for every positive integer $i$ and any $\epsilon >0$.
\end{enumerate}
\end{definition}

Product probability distributions $p_i = p^{\otimes n_i}$  where $p$ is a probability distribution over $\F_q$  admit a nice family of typical sets as shown by the 
the standard definition of a typical set recalled in Proposition \ref{Proposition:1}. It corresponds to $H=H_q(p)$, $K_1=K_2=1$, $\alpha(\epsilon)=\beta(\epsilon)=\epsilon$.
We also notice that for a sequence of probability distributions admitting a nice family of typical sets the asymptotic entropy per symbol $\lim_{i \rightarrow \infty} \frac{H_q(p_i)}{n_i}$ exists and is equal to $H$:
\begin{proposition}\label{prop:nice_typical}
For a nice family of typical sets, the quantity $H$ appearing in its definition is necessarily unique and the entropy per symbol $\frac{H_q(p_i)}{n_i}$ has a limit when
$i$ tends to infinity:
$$
\lim_{i \rightarrow \infty} \frac{H_q(p_i)}{n_i} = H.
$$
\end{proposition}
\begin{proof}
Fix $\eps > 0$ and a nonnegative integer $i$ and $ \delta = \delta(\eps,n_i)$.  From the definition of a typical set we show that the entropy $H_q(p_i)$ is bounded by
\begin{equation}\label{eq:sandwich}
 - (1-\delta) \log_q(B_i(\epsilon)) \leq H_q(p_i) \leq - (1-\delta) \log_q(A_i(\epsilon)) + \delta n_i - \delta \log \delta,
\end{equation}
The lower bound just comes from the lower bound
$$H_q(p_i) \geq - \sum_{\xv \in \Typc{n_i}} \log_q(p_i(\xv)) p_i(\xv) \geq - \log_q(B_i(\epsilon)) p_i\left( \Typc{n_i} \right)=- (1-\delta) \log_q(B_i(\epsilon)).$$
The upper bound follows from
\begin{eqnarray*}
H_q(p_i) &= & - \sum_{\xv \in \Typc{n_i}} \log_q(p_i(\xv)) p_i(\xv) - \sum_{\xv \in \F_q^{n_i}\setminus\Typc{n_i}} \log_q(p_i(\xv)) p_i(\xv)\\
& \leq & - (1-\delta) \log_q(A_i(\epsilon)) + \delta \sum_{\xv \in \F_q^{n_i}\setminus\Typc{n_i}} - \log_q(p_i(\xv)) q_i(\xv)\;\;\text{(where $q_i(\xv)=\frac{p_i(\xv)}{\delta}$)}\\
& \leq & - (1-\delta) \log_q(A_i(\epsilon)) + \delta \left(\sum_{\xv \in \F_q^{n_i}\setminus\Typc{n_i}} - \log_q(q_i(\xv)) q_i(\xv)\right) - \delta \log_q(\delta)\\
& \leq & - (1-\delta) \log_q(A_i(\epsilon)) + \delta n_i - \delta \log_q(\delta).
\end{eqnarray*}
The last inequality comes from the fact that $q_i$ can be interpreted as a probability distribution over $\F_q^{n_i}\setminus\Typc{n_i}$ and therefore
$- \sum_{\xv \in \F_q^{n_i}\setminus\Typc{n_i}} \log_q(q_i(\xv)) q_i(\xv) = H_q(q_i)$ and by using that
the $q$-ary entropy of a random variable taking $M$ different values is upper-bounded by $\log_q(M)$.

In the case of a nice typical set $\Typc{n_i}$, \eqref{eq:sandwich} gives 
$$(1-\negl(i))\left( H-\beta(\epsilon) - \frac{\log_q K_2}{n_i}\right) \leq \frac{H_q(p_i)}{n_i} \leq (1-\negl(i))\left( H+\alpha(\epsilon) - \frac{\log_q K_1}{n_i}\right) + \negl(i).$$
Letting $\epsilon$ going to zero shows that
$$
(1-\negl(i))\left( H - \frac{\log_q K_2}{n_i}\right) \leq \frac{H_q(p_i)}{n_i} \leq (1-\negl(i))\left( H - \frac{\log_q K_1}{n_i}\right) + \negl(i).
$$
Then letting $i$ go to infinity shows that $\frac{H_q(p_i)}{n_i}$ has a limit equal to $H$.
\end{proof}

Another way to put this definition is to say that a sequence of probability distributions admitting a nice family of typical sets satisfies some form of the asymptotic equipartition property (AEP):
almost all the mass of the probability distribution $p_i$ concentrates on $\Typc{n_i}$ for every $\epsilon>0$ as $i$ tends to infinity, these sets have size $\approx q^{n_iH}$  and the probabilities of the elements in this subset
is close to $q^{n_i H}$ up to a multiplicative constant of the form $q^{n_i\gamma(\epsilon)}$ where $\gamma(\epsilon)$ tends to zero with $\epsilon$.

\subsection{Achievability Result}

We are going to prove here the following theorem:
\tractability*

As we would like to prove a similar result for the rank metric later on, we are going to prove a more general result here, which will enable us to prove both results. 
\begin{theorem}\label{thm:general}
	Let $q \ge 2$ be a prime power. We assume that we have an increasing sequence of positive integers
        $(n_i)_{i \geq 1}$ and a family of functions $(f_i)_{i \geq 1}$ with $f_i : \F_q^{n_i} \rightarrow \mathbb{C}$ with $\norm{f_i}_2 = 1$  such that
        the sequence $(p_i)_{i \geq 1}$ defined by $p_i \eqdef |\QFT{f_i}|^2$ admits a nice family of typical sets. We let
        $H \eqdef \lim_{i \rightarrow \infty}\frac{H_q(p_i)}{n_i}$.
For every $R \in (0,H)$ the PGM solves $\QDP(q,n_i,\lfloor Rn_i \rfloor, f_i)$ with probability $1 - \negl(i)$.
\end{theorem}

Fix $R \in (0,H)$ and let $k_i = \lfloor Rn_i \rfloor$. From Proposition \ref{prop:PPGM}, we know that the average probability $\overline{\PPGM}$ over the randomness of codes (obtained by choosing a generator matrix $\Gm$ of $\C$ at random) is given by
\begin{eqnarray}
\overline{\PPGM} & = & \frac{1}{q^{k_i}} \esp_\Gm \lp\sum\limits_{\sv \in \F_q^{n_i}/\C^\perp }Z_\sv\rp^2\quad \text{where}\\
Z_\sv & \eqdef & \sqrt{\sum_{\ev \in \sv + \C^\perp}\left|\QFT{f_i}(\ev)\right|^2}.
\end{eqnarray}
We are going to show that $\overline{\PPGM}$ is equal to $1 - \negl(i)$. This implies the theorem on the spot. To prove this, we will consider
a random variable $\widetilde{Z_\sv}(\epsilon)$ defined for every $\epsilon >0$ which is related to $Z_\sv$ which is such that $Z_\sv \geq \widetilde{Z_\sv}(\epsilon)$,
but $Z_\sv \approx \widetilde{Z_\sv}(\epsilon)$ and $\widetilde{Z_\sv}(\epsilon)$ has a much smaller variance than $Z_\sv$, which allows us to prove concentration by 
using the second moment method.
This random variable is defined by
$$
\widetilde{Z_\sv}(\epsilon) \eqdef \sqrt{\sum_{\yv \in \Typi\cap (\sv + \CCp)}\left|\QFT{f_i}(\yv)\right|^2}
$$

\begin{lemma}\label{lem:EZs}
For any $\epsilon >0$, any positive integer $i$ and $\sv \in \F_q^{n_i}$, we have $\esp_\Gm\left\{\widetilde{Z_\sv}(\epsilon)^2\right\} \geq \frac{1-\delta(\epsilon,n_i)}{q^{k_i}}$
\end{lemma}
\begin{proof}
\begin{eqnarray}
	\esp_\Gm\left\{\widetilde{Z_\sv}(\epsilon)^2\right\} &= &\esp_\Gm \left\{\sum_{\yv \in \Typi\cap (\sv + \CCp)}p_i(\yv) \right\} \nonumber\\
	& = &  \sum_{\yv \in \Typi}p(\yv) \Pr_\Gm(\yv \in (\sv + \CCp) ) \nonumber \\
	&\ge &\frac{1}{q^{k_i}}\sum_{\yv \in \Typi}p(\yv) \label{eq:exception} \\
	&= &\frac{1-\delta(\epsilon,n_i)}{q^{k_i}}\text{ (by using the definition of $\delta(\epsilon,n_i)$).} \nonumber
\end{eqnarray}
Inequality \eqref{eq:exception} follows from the fact that $\Pr_\Gm(\yv \in (\sv + \CCp) )=1/q^{k_i}$ when $\yv \neq \sv$ and is
equal to $1$ if $\yv=\sv$.
\end{proof}

\begin{remk}\label{rmk:expectation}
Notice that the only reason why there is an inequality in the expression for $\esp\left\{\widetilde{Z_\sv}(\epsilon)^2\right\}$ is related to the possibility
that $\sv$ belongs to $\Typi$. If $\sv$ does not belong to $\Typi$ we have an equality instead and in any case if we let 
$\sT \eqdef \Typi \setminus \{\sv\}$ we have
\begin{equation}\label{eq:equalityinexpectation}
\esp\left\{\sum_{\yv \in \sT \cap (\sv + \CCp)}\left|\QFT{f_i}(\yv)\right|^2\right\}  =  \left\{ 
\begin{array}{ll} \frac{1- \delta(\epsilon,n_i)}{q^{k_i}} & \text{if $\sv \notin \Typi$} \\
\frac{1- \delta(\epsilon,n_i) - p(\sv)}{q^{k_i}} & \text{if $\sv \in \Typi$}
\end{array}
\right.
\end{equation}
\end{remk}
The variance of $\widetilde{Z_\sv}(\epsilon)^2$ on the other hand is upper-bounded by
\begin{lemma}\label{lem:varZs}
For any $\epsilon >0$, any positive integer $i$ and $\sv \in \F_q^{n_i}$, we have with the assumptions of Theorem
\ref{thm:general}
$$\Var\left[ \widetilde{Z_\sv}(\epsilon)^2 \right] \leq \frac{qK_2^2}{K_1}\cdot \frac{q^{-n_i(H - 2 \beta(\epsilon)-\alpha(\epsilon))}}{q^{k_i}},$$
where $K_1$, $K_2$ and $\alpha(\epsilon)$, $\beta(\epsilon)$ are respectively the constants and the functions appearing in Definition \ref{def:typical} of a nice family of typical sets.
\end{lemma}
\begin{proof}

	\begin{align*}
		\Var\left[ \widetilde{Z_\sv}(\epsilon)^2 \right] &= \Var\lp\sum_{\yv \in \Typi\cap (\sv+\CCp)}p(\yv)\rp\\
		&= \sum_{\yv \in \Typi}(p(\yv))^2 \Var(\mathds{1}_{\yv \in (\sv+\CCp)})  +  \sum_{\xv, \yv \in \Typi \atop \xv \neq \yv} p(\xv)p(\yv) \text{Cov}( \mathds{1}_{\xv \in (\sv+\CCp)},  \mathds{1}_{\yv \in (\sv+\CCp)})
	\end{align*}
	Here, 
	\begin{eqnarray*}
	\Var(\mathds{1}_{\yv \in (\sv+\CCp)})& = &\frac{1}{q^{k_i}}\lp 1- \frac{1}{q^{k_i}}\rp  \quad \text{if $\yv \neq \sv$}\\
	&=& 0 \quad \text{else.}
	\end{eqnarray*}
	When either $\xv=\sv$ or  $\yv =\sv$:
	$$
	\text{Cov}( \mathds{1}_{\xv \in (\sv+\CCp)},  \mathds{1}_{\yv \in (\sv+\CCp)})  = 0.
	$$
	When $\xv-\sv$ and $\yv -\sv$ are colinear (but none of them is zero), 
	\begin{align*}
		\text{Cov}( \mathds{1}_{\xv \in (\sv+\CCp)},  \mathds{1}_{\yv \in (\sv+\CCp)}) &= 
		\Pr(\xv \in (\sv+\CCp), \yv \in (\sv+\CCp)) -\Pr(\xv \in (\sv+\CCp)) \Pr(\yv \in (\sv+\CCp)) \\
		&= \frac{1}{q^{k_i}} - \frac{1}{q^{2k_i}}\\
		&= \frac{1}{q^{k_i}}\lp 1 - \frac{1}{q^{k_i}} \rp
	\end{align*}
	and when $\xv$ and $\yv$ are not colinear 
	\begin{align*}
\text{Cov}( \mathds{1}_{\xv \in (\sv+\CCp)},  \mathds{1}_{\yv \in (\sv+\CCp)}) &= 
\Pr(\xv \in (\sv+\CCp), \yv \in (\sv+\CCp)) -\Pr(\xv \in (\sv+\CCp)) \Pr(\yv \in (\sv+\CCp)) \\
&=\frac{1}{q^{2k_i}} - \frac{1}{q^{2k_i}} \\
&= 0
\end{align*}

	So the variance can be upper-bounded by
	\begin{align}
		\Var\left[ \widetilde{Z_\sv}(\epsilon)^2 \right]&\leq \frac{1}{q^{k_i}}\lp 1 - \frac{1}{q^{k_i}} \rp \lp \sum_{\yv \in \Typi}(p(\yv))^2 +  \sum_{\xv, \yv \in \Typi \atop \xv \neq \yv \text{ with } \xv \text{ and } \yv \text{ colinear}} p(\xv)p(\yv)\rp \nonumber \\
		&= \frac{1}{q^{k_i}}\lp 1 - \frac{1}{q^{k_i}} \rp    \sum_{\xv, \yv \in \Typi \atop \xv \text{ and } \yv \text{ colinear}} p(\xv)p(\yv) \nonumber \\
		&= \frac{1}{q^{k_i}}\lp 1 - \frac{1}{q^{k_i}}\rp\sum_{\yv \in \Typi}\lp \sum_{\xv \in  \Typi \atop \xv, \yv \text{ colinear}}p(\xv)p(\yv)\rp \nonumber \\
		&\le \frac{\abs{\Typi}}{q^{k_i}} qK_2^2 q^{-2n_i(H-\beta(\epsilon))} \nonumber \\
                & \le \frac{qK_2^2}{K_1}\frac{q^{-n_i(H - 2 \beta(\epsilon)-\alpha(\epsilon))}}{q^{k_i}},
\end{align}
where in the last inequality we used Lemma \ref{lem:cardinality} together with the definition of a nice family of typical sets
$$
\abs{\Typi} \leq \frac{1}{A_i(\epsilon)} \leq \frac{q^{n_i (H+\alpha(\epsilon))}}{K_1}.
$$
\end{proof}

These two lemmas yield the following concentration result
\begin{lemma}\label{lem:concentrationZseps}
Let $\gamma \eqdef H - R$. If $\epsilon$ is such that $2\beta(\epsilon)+\alpha(\epsilon) \leq \gamma/2$, then
$$
\Pr_\Gm\lc \widetilde{Z_\sv}(\eps)^2 \ge \frac{1}{q^{k_i}} (1 -\delta(\epsilon,n_i)- q^{-\frac{\gamma n_i}{8}})\rc = 1 - \negl(i).
$$
\end{lemma}
\begin{proof}
We use Bienaym\'e-Tchebychev's inequality for $X \eqdef \widetilde{Z_\sv}(\epsilon)^2$ and obtain for $\Delta >0$:
	\begin{equation}\label{eq:BT}
        \Pr\lp \abs{X - \esp(X)} \ge \Delta \rp \le \frac{\Var(X)}{\Delta^2}.
        \end{equation}
        To simplify expressions denote by $K$ the constant $K \eqdef \frac{qK_2^2}{K_1}$.  
        Let us rewrite a little bit the upper bound we have on
        $\Var(X)$ derived from Lemma \ref{lem:varZs}. We have
        \begin{eqnarray}
        \Var(X)& \leq &\frac{qK_2^2}{K_1}\cdot \frac{q^{-n_i(H - 2 \beta(\epsilon)-\alpha(\epsilon))}}{q^{k_i}} \nonumber\\
        & = & K \frac{q^{-n_i(H - R -2 \beta(\epsilon)-\alpha(\epsilon))}}{q^{2k_i}} \nonumber\\
        & \leq & K \frac{q^{-\frac{\gamma n_i}{2}}}{q^{2k_i}}. \label{eq:varZs}
\end{eqnarray}
We choose
$$
\Delta = \frac{q^{-\frac{\gamma n_i}{8}}}{q^{k_i}}.
$$
By using \eqref{eq:varZs} we obtain
\begin{eqnarray*}
\frac{\Var(X)}{\Delta^2} & \leq &  K \cdot \frac{q^{2k_i} \cdot q^{\frac{\gamma n_i}{4}}  q^{-\frac{\gamma n_i}{2}}}{q^{2k_i}} \\
& = & K \cdot q^{-\frac{\gamma n_i}{4}} \\
& = & \negl(i).
\end{eqnarray*}
We obtain by using this inequality
in \eqref{eq:BT} and replacing $\Delta$ by its value:
$$
\Pr_\Gm\lc \widetilde{Z_\sv}(\eps)^2 \ge \frac{1}{q^{k_i}} (1- \delta(\epsilon,n_i) - q^{-\frac{\gamma n_i}{8}})\rc = 1 - \negl(i).
$$
\end{proof}

A slight variation of this concentration result will be needed later on. It says that
\begin{lemma}\label{lem:concentrationZsepsprime}
Let $\gamma \eqdef H - R$. If $\epsilon$ is such that $2\beta(\epsilon)+\alpha(\epsilon) \leq \gamma/2$, then if we let $\sT \eqdef \Typi \setminus \{\sv\}$ we have
$$
\Pr_\Gm\lc \left|\sum_{\yv \in \sT \cap (\sv + \CCp)}\left|\QFT{f_i}(\yv)\right|^2 -E \right|\le \frac{q^{-\frac{\gamma n_i}{8}}}{q^{k_i}} \rc = 1 - \negl(i),
$$
where $E$ is the expected value of $\sum_{\yv \in \sT \cap (\sv + \CCp)}\left|\QFT{f_i}(\yv)\right|^2$. It is given by the following expression
$$
E = \left\{\begin{array}{ll} \frac{1- \delta(\epsilon,n_i)}{q^{k_i}} & \text{if $\sv \notin \Typi$} \\
\frac{1- \delta(\epsilon,n_i) - p(\sv)}{q^{k_i}} & \text{if $\sv \in \Typi$}
\end{array}
\right.
$$
\end{lemma}
\begin{proof}
We use for the expected value of $\sum_{\yv \in \sT \cap (\sv + \CCp)}\left|\QFT{f_i}(\yv)\right|^2$ its expression given in \eqref{eq:equalityinexpectation}.
The upper-bound on the variance of $\sum_{\yv \in \sT \cap (\sv + \CCp)}\left|\QFT{f_i}(\yv)\right|^2$ is clearly the same as the upper-bound on the variance of 
$\widetilde{Z_\sv}(\eps)^2$ given in Lemma \ref{lem:varZs}. This yields immediately Lemma \ref{lem:concentrationZsepsprime} similarly to what was done for obtaining 
Lemma \ref{lem:concentrationZseps} in the lower bound of expectation of the random variable $\sum_{\yv \in \sT \cap (\sv + \CCp)}\left|\QFT{f_i}(\yv)\right|^2$ to which we apply Bienaym\'e-Tchebychev's inequality. 
\end{proof}

We are ready now to prove Theorem \ref{thm:general}.
\begin{proof}[Proof of Theorem \ref{thm:general}]
Since $Z_\sv^2 \geq \widetilde{Z_\sv}(\epsilon)^2$ for any $\epsilon >0$, we obtain by using Lemma \ref{lem:concentrationZseps}
by choosing $\epsilon$ such that $2\beta(\epsilon)+\alpha(\epsilon) \leq \gamma/2$,
$$\Pr_\Gm\lc Z_\sv^2 \ge \frac{1 - \delta(\epsilon,n_i) - q^{-\frac{\gamma n_i}{8}}}{q^{k_i}} \rc = 1 - \negl(i)$$ which gives  $ \Pr_\Gm\lc Z_\sv\ge \sqrt{\frac{1}{q^{k_i}} (1 - \delta(\epsilon,n_i) - q^{-\frac{\gamma n_i}{8}})} \rc = 1 - \negl(i).$
	Since $\delta(\epsilon,n_i)=\negl(i)$, this implies for all $ \sv \in \Fqni/\C^\perp$ :
	$$	\esp_{\Gm}\lc Z_\sv\rc \ge  \frac{1}{\sqrt{q^{k_i}}} (1 - \negl(i)) $$
	Then, $	 \esp_{\Gm}\left[\sum\limits_{\sv \in \Fqni/\C^\perp} Z_\sv\right] \ge \sqrt{q^{k_i}}\left(1 - \negl(i)\right).$

	We use Jensen's inequality $\esp_{\Gm}(X^2) \ge (\esp_{\Gm}(X))^2$ in order to obtain:	
	$$\overline{\PPGM}=\frac{1}{q^{k_i}}\left(\esp_\Gm \left\{ \sum_{\sv \in \Fqni/\C^\perp}Z_\sv \right\}^2 \right)  \ge \frac{1}{q^{k_i}}\lp \esp\left\{\sum_{\sv \in \Fqni/\C^\perp}Z_\sv\right\}\rp ^2 \ge 1 - \negl(i).$$
\end{proof}

With these results at hand we can now prove Theorem \ref{Theorem:Tractability}.
\begin{proof}[Proof of Theorem \ref{Theorem:Tractability}]
In this case, $f_i = g^{\otimes i}$ and we
take $n_i=i$, $H= H_q(\abs{\QFT{g}}^2)$, $A_i(\epsilon)= q^{-i(H+\eps)}$, $B_i(\epsilon)= q^{-i(H-\eps)}$. Proposition \ref{Proposition:1} shows that
$\delta(\eps,n)=\negl(n)$ for all $\eps >0$. This shows that we can apply Theorem \ref{thm:general} to this case, from which we derive immediately Theorem \ref{Theorem:Tractability}.
\end{proof}

\subsection{Strong converse}

We are now going to prove that whatever code we take in the $\QDP$ problem, be it random as in the statement of the QDP problem or chosen carefully and  be it even non-linear, then we cannot solve the QDP problem with non negligible probability as soon as the rate is 
just slightly above $H_q(|\fg|^2)$. More precisely, we are going to prove

\intractability*

As for achievability results, we would like to prove a similar result for the rank metric later on. Then, we are going to prove a more general result here, which will enable us to prove both results :

\begin{theorem}\label{thm:nonachievability_general}

	We assume that we have an increasing sequence of positive integers $(n_i)_{i \geq 1}$ and a family of functions $(f_i)_{i \geq 1}$ with $f_i : \F_q^{n_i} \rightarrow \mathbb{C}$ with $\norm{f_i}_2 = 1$
such that the sequence $(p_i)_{i \geq 1}$ defined by $p_i \eqdef |\QFT{f_i}|^2$ admits a nice family of typical sets in the strong sense.

Let $R$ be a fixed constant in $(0,1)$ such that $R > H$. For any quantum algorithm and any code $\CC_i \in \F_q^{n_i}$ with at least $q^{Rn_i}$ codewords the probability $\psucc$ to solve $\QDP(\CC_i,f_i)$ satisfies
$$ \psucc  = \negl(i).$$
\end{theorem}

The proof of this statement relies on several statements. The first one uses the fact that for any $\cv$ in $\CC_i$ the state $\QFT{\ket{\psi_{\cv}}}$ which is equal to $\sum_{\yv \in \F_q^{n_i}} \hf_i(\yv) \chi_{\cv}(\yv) \ket{\yv}$ by Lemma
\ref{lem:qft_c} is extremely close to the state $\ket{\wpsi_\cv}$ which is obtained from $\QFT{\ket{\psi_{\cv}}}$ by keeping in the previous sum only those corresponding to a $\yv$ that are in the typical set $\Typcc{n_i}{\epsilon(i)}$ associated to $|\ff_i|^2$. 
\begin{lemma}\label{lem:candctilde}
Let $i$ be a positive integer. For $\cv \in \CC$, we let $\ket{\wpsi_\cv} \eqdef \frac{1}{\sqrt{1-\delta(\epsilon(i),n_i)}} \sum_{\yv \in \Typcc{n_i}{\eps(i)}} \hf_i(\yv) \chi_{\cv}(\yv) \ket{\yv}$, where $\delta(\eps(i),n_i) \eqdef 1 - \sum_{\yv \in \Typcc{n_i}{\eps(i)}} |\ff_i(\yv)|^2$. These states are normalized and 
we have $\braket{\QFT{\psi}_\cv}{\wpsi_\cv} =   \sqrt{1-\delta(\epsilon(i),n_i)}$.
\end{lemma}
\begin{proof}
This is a direct consequence of
$$\norm{\ket{\wpsi_\cv}}_2^2 = \frac{1}{1-\delta(\epsilon(i),n_i)} \sum_{\yv \in \Typcc{n_i}{\eps(i)}} |\hf_i(\yv)|^2 =1.$$
The second statement is proved in a similar fashion
$$
\braket{\QFT{\psi}_\cv}{\wpsi_\cv}  =  \frac{\sum_{\yv \in \Typcc{n_i}{\eps(i)}} |\ff(\yv)^2|}{\sqrt{1-\delta(\epsilon(i),n_i)}} = \frac{1-\delta(\epsilon(i),n_i)}{\sqrt{1-\delta(\epsilon(i),n_i)}}=\sqrt{1-\delta(\epsilon(i),n_i)}.
$$
\end{proof}
The point is that these states $\ket{\wpsi_\cv}$ live in a much smaller space which is of dimension the size of the typical set. There is a general upper bound on how we can distinguish such states which is given by
\begin{lemma}\label{lem:smallerspace}
Assume that we have $N$ states living in a space of dimension $K$. The average probability of distinguishing them is upper bounded by $\frac{K}{N}$.
\end{lemma}
\begin{proof}
Let $\ket{\phi_1},\cdots,\ket{\phi_N}$ be those states and let $M_1,\cdots,M_N$ be a POVM distinguishing those states. The average probability of success $\psucc$ for this POVM is given by
$$
\psucc = \frac{1}{N} \sum_{i=1}^N \tr\left(M_i \kb{\phi_i}\right).
$$
Now let $V$ be the space of dimension $K$ containing the $\ket{\phi_i}$'s  and let $\Pi_V$ be the orthogonal projector onto $V$, that is $\Pi_{V} \eqdef \sum_{i=1}^K \kb{\ev_i}$ where $\{\ket{\ev_1},\cdots,\ket{\ev_K}\}$ is an orthonormal basis
of $V$. Since each $\ket{\phi_i}$ has support in $V$, all these states satisfy
\begin{equation}
\label{eq:domination}
 \kb{\phi_i} \preceq \Pi_V.
\end{equation}
By using Claim \ref{Claim:Trace} with \eqref{eq:domination} we obtain
\begin{eqnarray*}
\psucc & = & \frac{1}{N} \sum_{i=1}^N \tr\left(M_i \kb{\phi_i}\right)\\
& \leq & \frac{1}{N} \sum_{i=1}^N \tr\left(M_i \Pi_V\right) \text{ (by Claim \ref{Claim:Trace} and \eqref{eq:domination})}\\
& = & \frac{1}{N} \tr\left(\sum_{i=1}^N M_{i} \Pi_{V}\right) \\
& = & \frac{1}{N} \tr(\Pi_V) \text{ (since $\sum_{i=1}^N M_i = I$)}\\
& = & \frac{K}{N} \text{ (since $\tr(\Pi_{V}) = K$)}.
\end{eqnarray*}
\end{proof}

The last ingredient tells us that we can not hope to have a much better success probability if instead of feeding our distinguisher with the $\ket{\wpsi_\cv}$ we use the $\ket{\psi_\cv}$ since we have in general
\begin{lemma}\label{lem:approximate}
Let $\ket{\phi_1},\cdots,\ket{\phi_N}$ and $\ket{\wphi_1},\cdots,\ket{\wphi_N}$ be $2N$ quantum states all living in the same Hilbert space. Let
$$
\delta \eqdef \max_{i \in \{1,\cdots,N\}} \sqrt{1- |\braket{\phi_i}{\wphi_i}|^2} = \max_{i \in \{1,\cdots,N\}} D\left(\kb{\phi_i},\kb{\wphi_i}\right).
$$
Let $p$ be the average success probability of
distinguishing the $\ket{\phi_i}$'s by a certain POVM $\{M_1,\cdots,M_N\}$. Let $\widetilde{p}$ be the average success probability of
distinguishing the $\ket{\wphi_i}$'s by the same POVM. We have
$$
p \leq \widetilde{p} + \delta.
$$
\end{lemma}
\begin{proof}
	We write 
	\begin{align*}
		p & = \frac{1}{N} \sum_{i=1}^N \tr\left(M_i \kb{\phi_i}\right) \\
		& \le \frac{1}{N} \sum_{i=1}^N \tr\left(M_i \kb{\wphi_i}\right) + \delta & \text{(by using Claim \ref{Claim:TD})} \\
		& = \widetilde{p} + \delta.
		\end{align*}
\end{proof}
We are ready now to prove Theorem \ref{thm:nonachievability_general}.
\begin{proof}[Proof of Theorem \ref{thm:nonachievability_general}]
 We recall that for any positive integer $i$
	\begin{enumerate}
		\item $\sum_{\yv \in \Typcc{n_i}{\epsilon(i)}} |\hf_i(\yv)|^2 = \sum_{\yv \in \Typcc{n_i}{\epsilon(i)}} |\hf_i(\yv)\chi_{\cv}(\yv)|^2 = 1 - \delta(\epsilon(i),n_i)$ with $\delta(\epsilon(i),n_i) = \negl(i)$.
		\item $|\Typcc{n_i}{\epsilon(i)}| \le q^{n_i(H+ o(1))}$.
	\end{enumerate}
        For the last point, we used Lemma \ref{lem:cardinality} which says that $|\Typcc{n_i}{\epsilon(i)}| \le \frac{1}{A_i(\epsilon(i))}$ and the fact that $A_i(\epsilon(i)) = q^{-n_i(H+o(1))}$.
Now, consider a quantum POVM $\{M_\cv\}_{\cv \in \C_i}$. Let $\psucc$ be the average probability to solve the $\QDP$ problem using this measurement
        whereas $\widetilde{p}$ is the average probability of distinguishing $\cv$ with the same POVM when fed instead of $\ket{\QFT{\psi}_\cv}$ the state
        $\ket{\wpsi_\cv}$. By using Lemma \ref{lem:approximate} we deduce that
\begin{equation}
\label{eq:almost_done}
\psucc \le \tp + \sqrt{\delta(\epsilon(i),n_i)},
\end{equation}
since $\braket{\QFT{\psi}_\cv}{\wpsi_\cv} =   \sqrt{1-\delta(\epsilon(i),n_i)}$ for every $\cv \in \CC$ by Lemma \ref{lem:candctilde}.
By applying Lemma \ref{lem:smallerspace} to the $q^{Rn_i}$ states $\ket{\wphi_\cv}$ for $\cv \in \CC_i$ we obtain
\begin{eqnarray*}
\widetilde{p} & \leq &\frac{\left|\Typcc{n_i}{\epsilon(i)}\right|}{q^{Rn_i}} \\
&\leq & \frac{q^{n_i(H+o(1))}}{q^{Rn_i}} \\
& \leq & q^{n_i(H-R+o(1))}\\
& = & \negl(i).
\end{eqnarray*}
By using this upper-bound in \eqref{eq:almost_done} we obtain
$$
p \leq \negl(i) + \sqrt{\delta(\epsilon(i),n_i)} = \negl(i),
$$
where we used Proposition \ref{Proposition:1} in the last equality  to write that $\delta(\epsilon(i),n_i) = \negl(i)$.
\end{proof}

Theorem \ref{thm:nonachievability} is a direct corollary of this theorem. Indeed in this case we take
\begin{enumerate}
\item $n_i = i$ and $f_i = g^{\otimes i}$.
\item $H = H_q(|\QFT{g}|^2)$
\item The typical set is defined with the bounding functions
$A_i(\epsilon) = q^{-i(H+\epsilon)}$ and $B_i(\epsilon)=q^{-i(H-\epsilon)}$.
\end{enumerate}
We then choose $\epsilon$ as a function of $i$ of the form $\epsilon(i)=\frac{1}{i^{1/3}}$ and by applying Proposition \ref{Proposition:1} we verify that the defect function
$\delta(\epsilon(i),n_i)$ satisfies $\delta(\epsilon(i),n_i)=\negl(i)$.

\subsection{Application to the rank metric}
 In all this subsection, the codelength $n$ is of the form
 $n=a\cdot b$ for some integers $a$ and $b$ and we view elements of $\F_q^{n}$ as matrices in $\F_q^{a \times b}$ as explained in \S \ref{ss:notation}.
 We will assume that $a \geq b$. If this is not the case, we will just swap the role of rows and columns in the explanations that follow.
 Unlike the error model we considered in the previous subsections, the rank metric case corresponds in a natural way to an error model which is not memoryless.
 There are several ways of defining an error distribution on $\F_q^n$ which would make sense for cryptographic applications. A uniform
 distribution over all errors of $\F_q^n$ of a given rank weight $t$ would for instance be relevant here. This would complicate the computations of the Fourier transform
 of the error distribution. A simple way to approximate this distribution is given in \cite[\S III.E]{DRT24}. It corresponds to an error amplitude distribution $f_t^{a,b}$ defined by
 \begin{eqnarray}
   \sum_{\ev \in \F_q^n} f_t^{a,b}(\ev) \ket{\ev} & \eqdef & \frac{1}{\sqrt{Z}} \sum_{\substack{V:V \leq \F_q^b \\\dim V =t}} \ket{\pi_V}\;\;\text{where} \label{eq:rank_state}\\
     \ket{\pi_V} & \eqdef & \left(\frac{1}{\sqrt{q^{\dim V}}}\sum_{\vv \in V} \ket{v}\right)^{\otimes a}, \label{eq:rank_state_V}
   \end{eqnarray}
 and $Z$ is a normalizing constant which makes the state appearing on the right of \eqref{eq:rank_state} to be a valid quantum state \ie\ a state of norm $1$. The state
 $\ket{\pi_V}$ can be seen as a uniform superposition of vectors $\xv \in \F_q^n$ whose matrix form $\matv(\yv)$ gives all matrices formed by rows taken from the
 vector subspace $V$. As shown in \cite[Lemma 3.21]{DRT24} $f(\ev)$ is radial, meaning that it only depends on the rank weight of $\ev$ and we have
 $$
 f_t^{a,b}(\ev) = \left\{
 \begin{array}{ll}
	\dfrac{\left[
		\begin{array}{c}
			b-|\ev|_\rk \\
			t-|\ev|_\rk
		\end{array}
		\right]_q
	}{\sqrt{q^{at}Z}} & \mbox{if } \abs{\ev}_\rk \leqslant t \\
	0 & \mbox{else }
	
\end{array}
\right.
$$
with $\left[
\begin{array}{c}
	b \\
	t
\end{array}
\right]_q = \left\{
\begin{array}{ll}
	\prod\limits_{i = 0}^{t-1}\dfrac{q^b-q^i}{q^t-q^i}
	 & \mbox{if } t \leqslant b \\
	0 & \mbox{else }
	
\end{array}
\right.$
being the Gaussian binomial coefficient and the normalizing constant is shown to satisfy (see \cite[Lemma 3.22]{DRT24}) $Z = \Theta\lp\left[
\begin{array}{c}
	b \\
	t
\end{array}
\right]_q \rp$.
We will generally drop the dependence in $a, b$ which will be clear from the context and simply write $f_t(e)$.
\\
Moreover the Gaussian binomial coefficient and the size $S_t$ of the sphere of radius $t$ for the rank metric over $\F_q^{n}$
(that is the vectors $\yv$ in $\F_q^n$ such that $\matv(\yv)$ is of rank $t$) satisfy
\begin{fact} \cite[Section II.B]{DRT24}  \label{fa:DRT24}
$\left[\begin{array}{c}
	b \\
	t
\end{array}
    \right]_q = \Theta(q^{t(b-t)})$ , $S_t = \Theta\lp q^{t(a+b-t)}\rp$, $|f_t(\ev)|^2 = \Th{q^{(b-t)(t-2u)-at}}$.
  The constants hidden inside $\Theta$ can be taken to be absolute constants that do not depend on $a$, $b$ or $t$.
\end{fact}

It follows directly from \cite[Prop. 3.23]{DRT24} that
\begin{equation}\label{eq:QFT_rank}
  \QFT{f_t} = f_{b-t}.
\end{equation}

The function $f_t$ is radial, meaning that the value  $f_t(\ev)$ only depends on the rank weight $|\ev|_\rk$ of $\ev$. 
We will be interested in the dual distribution $p \eqdef \left| \QFT{f}_t\right|^2= \left| f_{b-t} \right|^2$. This probability distribution  concentrates around the rank weight $b-t$ as shown in Appendix B.2 of \cite{DRT24}
\begin{fact}
\label{fa:concentration}
For any $u$ in $\{0,\cdots,b-t\}$ we have
 \begin{eqnarray}
   \tilde{p}(u) &= & \Th{q^{-(a+t)(b-t)+2tv}} \label{eq:ptilde1}\\
   \sum_{\ev : |\ev|_\rk = u} p(\ev) &= &\tilde{p}(u)S_u = \Theta\lp q^{v(b-a)-v^2} \rp \label{eq:concentration}
  \end{eqnarray}
where $v \eqdef b-t-u$ and $\tilde{p}(|\ev|_\rk) \eqdef p(\ev)$ for any $\ev \in \F_q^n$, where
the constant hidden in $\Theta$ notation is an absolute constant that does not depend on $a$, $b$ or $q$. 
\end{fact}
The tails of $p$ are readily seen to verify
\begin{lemma}\label{lem:tails}
  Let $\epsilon >0$.
  $$
\sum_{u=0}^{\lfloor(1-\epsilon)b\rfloor -t} S_u \tilde{p}(u) = \OO{q^{-b^2\varepsilon^2}}.
  $$
\end{lemma}
\begin{proof}
  We have
  \begin{eqnarray*}
    \sum_{u=0}^{\lfloor(1-\epsilon)b\rfloor -t} S_u \tilde{p}(u) & = & \Th{\sum_{v=\lceil b \epsilon \rceil}^{b-t}q^{v(b-a)-v^2}}\quad\text{(by Fact \ref{fa:concentration})}\\
        & = & \OO{q^{-b^2\varepsilon^2}\sum_{i=0}^{\infty} q^{-i}}\\
        & = & \OO{q^{-b^2 \varepsilon^2}}.       
  \end{eqnarray*}
\end{proof}  
The typical set is defined as
\paragraph{Definition and properties of the set $\Typ$.} We let
\begin{equation}\label{eq:typical}
  \Typ \eqdef \left\{ \ev \in \F_q^n: b(1-\epsilon) - t \leq |\ev|_\rk \leq b-t\right\}. 
\end{equation}
Here, $p(\ev)$ is clearly a decreasing function of the rank weight $|\ev|_\rk$. This shows that this is indeed a typical set as specified in Definition \ref{def:typical}.
From \eqref{eq:ptilde1} we infer that the bounding functions $A(\epsilon)$ and $B(\epsilon)$ satisfy
\begin{eqnarray}
  A(\epsilon) & = & \Th{q^{-nH}} \label{eq:A}\\
  B(\epsilon) & = & \Th{q^{-n\left(H - \frac{2t}{a}\epsilon\right)}}\quad \text{where} \label{eq:B}\\
    H & \eqdef & (1+t/a)(1-t/b).
\end{eqnarray}
From Lemma \ref{lem:tails} we deduce that
\begin{equation}\label{eq:delta}
\delta(\epsilon,n) = \OO{q^{-b^2 \epsilon^2}}.
\end{equation}

\begin{remk}
  The quantity $H$ can really be viewed as a very good approximation of the entropy per symbol for the distribution $p$. It is straightforward to use Fact \ref{fa:concentration} to show
  that the entropy per symbol $\tilde{H} \eqdef \frac{H_q(p)}{n}$ is of the form
  $\tilde{H} = H + \OO{1/a}$. 
  \begin{eqnarray*}
     H_q(p) & = & - \sum_{u=0}^{b -t} \log_q(\tilde{p}(u))  \tilde{p}(u) S_u \\
     & = & - \sum_{v=0}^{b -t} \left( -(a+t)(b-t)+2tv + \Th{1} \right) \tilde{p}(b-t-v)S_{b-t-v} \quad \text{(by \eqref{eq:ptilde1})}\\
     & = & (a+t)(b-t)  - \Th{1} - 2t \sum_{v=0}^{b -t} v \tilde{p}(b-t-v)S_{b-t-v} \quad \text{(since $\sum_{v=0}^{b-t} \tilde{p}(b-t-v) S_{b-t-v}=1$)}\\
     & = & (a+t)(b-t) + \OO{b} .
  \end{eqnarray*}
  The last point follows from $\sum_{v=0}^{b-t} v \tilde{p}(b-t-v) S_{b-t-v}= \OO{1}$ which can be derived by using Fact \ref{fa:concentration} similarly to what is done in Lemma \ref{lem:tails}.
\end{remk}
From there, we give the \textbf{achievability result}:

\begin{theorem}\label{thm:achievability_rm}
	Let $q \ge 2$ be a prime power. We assume that there exist three increasing sequences $(a_i)_{i \ge 1}$, $(b_i)_{i \ge 1}$, $(t_i)_{i \ge 1}$ such that $a_i \geq b_i$ for any positive integer $i$ and the 
	sequence  $(1+t_i/a_i)(1-t_i/b_i)$ tends to a constant $H \in (0,1)$.  
	
	For every $R \in (0,H)$ the PGM solves $\QDP(q,a_ib_i,\lfloor Ra_ib_i \rfloor, f_{t_i}^{a_i, b_i})$ with probability $1 - \negl(i)$.
\end{theorem}

\begin{proof}
From \eqref{eq:delta} it appears that
          $\delta(\epsilon,n_i) = \negl(i)$ for any $\eps >0$.	Using \eqref{eq:B}, we can take $\alpha(\eps)=0$, $\beta(\eps)=2\eps$ and the assumptions of Theorem \ref{thm:general} are therefore verified.
\end{proof}

\textbf{Non-achievability} follows also immediately from Theorem \ref{thm:nonachievability_general}
	\begin{theorem}\label{thm : inachievability_rm}
		Let $q \geqslant 2$ be a prime power. We assume that there exist three increasing sequences $(a_i)_{i \ge 1}$, $(b_i)_{i \ge 1}$, $(t_i)_{i \ge 1}$ such that $a_i \geq b_i$ for any positive integer $i$ and the 
		sequence  $(1+t_i/a_i)(1-t_i/b_i)$ tends to a constant $H \in (0,1)$. Let $R$ be a constant in $(H,1]$. Then for any quantum algorithm and any family of codes $(\C_i)_{i \ge 1}$ where $\C_i \in \F_q^{n_i}$ contains at least $q^{Rn_i}$ codewords with $n_i \eqdef a_i\cdot b_i$, the probability $\psucc$ to solve $\QDP(\CC_i,f^{a_i,b_i}_{t_i})$ satisfies
		$$ \psucc  = \negl(i).$$
	\end{theorem}
	
        \begin{proof} We define the gap function $\epsilon$ as $\epsilon(i) = \frac{1}{b_i}$. From \eqref{eq:delta} it appears that
          $\delta(\epsilon(i),n_i) = \negl(i)$. The fact that $A(\epsilon(i))=q^{-n_i(H+o(1))}$ as $i$ tends to infinity follows from \eqref{eq:A} and the fact that
          $(1+t_i/a_i)(1-t_i/b_i)$ tends to $H$.

\end{proof}

	\section{Sampling Dual Codewords}
       	The algorithm for finding short codewords given in \cite{CT24} applies strictly speaking to the Bernoulli noise and the Hamming metric.   However, the algorithm
detailed in the full version of the paper \cite[\S 6.1]{CT23} can be used with any noise distribution. It is based on a slight variation of Regev's quantum algorithm which allowed to reduce the SIS problem to LWE.
Roughly speaking, Regev's approach \cite{R05} revisited by Chen, Liu and Zhandry in \cite{CLZ22} consists in constructing a superposition of dual codewords $\sum_{\cv \in \CC^\perp} \hat{f}(\cv) \ket{\cv}$ by solving the quantum decoding problem for
$\C$ and noise amplitude function $f$. More precisely, it consists in first preparing the uniform superposition of codewords tensored with a superposition of errors
$\frac{1}{\sqrt{|\C|}} \sum_{\cv \in \C} \ket{\cv} \otimes \sum_{\ev \in \F_q^n} f(\ev)\ket{\ev}$
and adding the first register to the second one to create the entangled state $\frac{1}{\sqrt{|\C|}} \sum_{\cv \in \C, \ev \in \F_q^n} f(\ev) \ket{\cv}\ket{\cv + \ev}$.
This state can be rewritten as $\frac{1}{\sqrt{|\C|}} \sum_{\cv \in \CC} \ket{\cv}\ket{\psi_\cv}$ with $\ket{\psi_\cv} \eqdef \sum_{\ev \in \F_q^n} f(\ev) \ket{\cv + \ev}$.
  If we solve the quantum decoding problem for $\C$ and $f$, we can in principle recover $\cv$ and erase the first register to get
  $\sum_{\cv \in \C, \ev \in \F_q^n} f(\ev) \ket{\cv + \ev}$ in the second register. This state $\sum_{\yv \in \F_q^n} \alpha_\yv \ket{\yv}$ is periodic with periods the codewords of $\CC$. We mean here
  that $\alpha_\yv = \alpha_{\yv+\cv}$ for any $\yv \in \F_q^n$ and any $\cv \in \CC$. When applying the quantum Fourier transform to it, we obtain precisely the state we are after,
  namely up to a normalizing factor $\sum_{\cv \in \C^\perp} \hat{f}(\cv) \ket{\cv}$. The whole process is summarized by the following sequence of operations
  
  \noindent\fbox{\parbox{\textwidth}{
\begin{center}
{\bf Algorithm 4 : Regev's algorithm for sampling dual codewords.}\label{algo: sampling}
\end{center}
\begin{align*}
	&\text{Initial state preparation:} &  & \quad \ket{\phi_0}&=& \quad \dfrac{1}{\sqrt{\abs{\C}}}\sum_{\cv \in \C}\sum_{\ev \in \F_q^n}f(\ev)\ket{\cv}\ket{\ev} \\
  &\text{adding $\cv$ to $\ev$:} & \mapsto &  \quad \ket{\phi_1}&=& \quad \frac{1}{\sqrt{|\C|}} \sum_{\cv \in \C} \sum_{\ev \in \F_q^n}  f(\ev) \ket{\cv}\ket{\cv+\ev}\\
  & & & &= & \quad \frac{1}{\sqrt{|\C|}} \sum_{\cv \in \C} \ket{\cv}\ket{\psi_\cv}  \\
  &\text{erase the $1$st register by decoding $\ket{\psi_\cv}$:} & \mapsto  & \quad
  \ket{\phi_3}& =& \quad \frac{1}{\sqrt{|\C|}} \sum_{\cv \in \C} \ket{\zero}\ket{\psi_\cv} \\
  & & & & = & \quad \ket{\zero} \frac{1}{\sqrt{|\C|}} \sum_{\cv \in \C, \ev \in \F_q^n} f(\ev) \ket{\cv + \ev} \\
  &\text{apply the QFT:} & \mapsto & \quad \ket{\phi_4} &=&\quad \dfrac{1}{\sqrt{Z}} \sum_{\cv \in \C^\perp} \hat{f}(\cv) \ket{\cv}  \\
&\text{measuring the whole state:} & \mapsto & \quad  & & \quad \ket{\yv} \;\;\text{with    $\yv \in \C^\perp$.}
\end{align*}}}

\vspace{0.2cm} When we measure the resulting state we sample dual codewords with respect to a probability distribution proportional to $|\hat{f}|^2$ and supported in $\CC^\perp$.
This is the general picture, however using a decoding algorithm which is not perfect, meaning that from the knowledge of $\ket{\psi_\cv}$ we do not always recover $\cv$, adds some difficulties because we do not recover exactly $\frac{1}{\sqrt{|\C|}} \sum_{\cv \in \C, \ev \in \F_q^n} f(\ev) \ket{\cv + \ev}$ after the decoding step. The case of the Bernoulli noise was treated in \cite{CT24}, this allows to obtain  codewords of small Hamming weight. At the limit where the quantum decoding problem can still be solved, we get codewords at the Gilbert Varshamov distance, in other words typically minimum weight codewords. We are going to generalize this result to more general noise amplitude distributions $f$ and we will show that
\begin{itemize}
\item Regev's approach used in conjunction with the PGM for solving the quantum decoding problem allows to sample dual codewords according to a distribution supported  on $\CC^\perp \setminus \{ \zero \} $ and proportional to $|\hat{f}|^2$ on this support in the tractability regime;
  \item at the limit of the tractability regime, we essentially sample the most likely non zero dual codewords according to the aforementioned probability distribution.
\end{itemize}

\subsection{The sampling algorithm}
In the whole section $\CC$ denotes a linear code of length $n$  over $\Fq$ whose generator matrix $\Gm$ is drawn uniformly at random in $\F_q^{k \times n}$, meaning that the
code is of dimension $k$ with probability $1-o(1)$ for a fixed rate $R=\frac{k}{n}$ bounded away from $1$.
Roughly speaking, in order to avoid destroying the second register carrying $\ket{\psi_\cv}$ and still measuring most of the time $\cv$, we apply the PGM coherently on the second register carrying $\ket{\cv}$ and write the output on the third register. In our case applying the PGM coherently basically amounts to apply the following transformation
\begin{equation}\label{eq:PGM-Yc}
\ket{\QFT{\psi_\cv}}\ket{\zero} \stackrel{\text{PGM}}{\mapsto} \sum_{\cv' \in \C}\alpha_{\cv,\cv'}\ket{Y_{\cv'}}\ket{\cv'}
\end{equation}
where
$\ket{\QFT{\psi_\cv}} = \sum_{\cv' \in \C}\alpha_{\cv,\cv'}\ket{Y_{\cv'}}$ and the $\ket{Y_\cv}$ for $\cv$ in $\C$ form an orthonormal set of vectors
defined by
\begin{equation}
\label{eq:Yc}
\ket{Y_\cv} \eqdef \frac{1}{\sqrt{q^k}} \sum_{\sv \in \F_q^k} \chi_\cv(\uv_\sv) \ket{\tilde{W_\sv}}
\end{equation}
where $\uv_\sv$ is an arbitrary element of $\C_\sv^\perp \eqdef \{\xv \in \F_q^n:\Gm \trsp{\xv} = \trsp{\sv}\}$ and
\begin{eqnarray}
\ket{\tilde{W}_\sv}& \eqdef &\frac{\ket{W_\sv}}{Z_\sv} \text{ with} \label{eq:tWs}\\
\ket{W_\sv}& \eqdef & \sum_{\yv \in \C_\sv^\perp} \QFT{f}(\yv)\ket{\yv} \label{eq:Ws}\\
Z_\sv & \eqdef & \norm{Z_\sv}. \label{eq:Zs}
\end{eqnarray}
 It turns out, see Proposition 18 in \cite{CT23} that the PGM
associated to the ensemble of states $\{\QFT{\ket{\psi_\cv}}\}_{\cv \in \C}$ is nothing but the projective measurement $\{\ketbra{Y_\cv}{Y_\cv} \}_{\cv \in \C}$. This explains
\eqref{eq:PGM-Yc}.
The value of the third register is then subtracted to the first register and we reverse at that point the PGM between the second and third register. As we have seen in the previous section it is more convenient to express the PGM in this case in the Fourier basis. We apply therefore prior to applying the PGM the quantum Fourier transform on
the second register carrying $\ket{\psi_\cv}$. This is basically what is done in Algorithm of Section 6.3 in \cite{CT23}, which we now recall
\noindent\fbox{\parbox{\textwidth}{
\begin{center}
{\bf Algorithm 4.1 : Algorithm of the quantum reduction with PGM.}\label{algo: reductionPGM}
\end{center}
\begin{align*}
	&\text{Initial state preparation:} &  & \quad \ket{\phi_0} &=& \quad \dfrac{1}{\sqrt{\abs{\C}}}\sum_{\cv \in \C}\sum_{\ev \in \F_q^n}f(\ev)\ket{\cv}\ket{\ev}\ket{\zero} \\
  &\text{adding $\cv$ to $\ev$:} & \mapsto &  \quad \ket{\phi_1} &=& \quad \frac{1}{\sqrt{|\C|}} \sum_{\cv \in \C} \sum_{\ev \in \F_q^n}  f(\ev) \ket{\cv}\ket{\cv+\ev}\ket{\zero}\\
  & & & &=& \quad \frac{1}{\sqrt{|\C|}} \sum_{\cv \in \C} \ket{\cv}\ket{\psi_\cv}\ket{\zero}  \\
        & \text{apply the QFT on $\ket{\psi_\cv}$:} & \mapsto & \quad \ket{\phi_2} &=& \quad \dfrac{1}{\sqrt{\abs{\C}}}\sum_{\cv \in \C} \ket{\cv}\ket{\QFT{\psi_\cv}}\ket{\zero} \\
  &\text{applying the PGM coherently:} & \mapsto  & \quad
  \ket{\phi_3}& =& \quad \dfrac{1}{\sqrt{\abs{\C}}}\sum_{\cv \in \C} \ket{\cv}\sum_{\cv' \in \C}\alpha_{\cv,\cv'}\ket{Y_{\cv'}}\ket{\cv'}\\
  & & & & &\text{ with } \ket{\QFT{\psi_\cv}} = \sum_{\cv' \in \C}\alpha_{\cv,\cv'}\ket{Y_{\cv'}} \\
  &\text{erase $\cv$:} & \mapsto & \quad \ket{\phi_4} &=& \quad \dfrac{1}{\sqrt{\abs{\C}}}\sum_{\cv, \cv' \in \C}\alpha_{\cv,\cv'}\ket{\cv-\cv'}\ket{Y_{\cv'}}\ket{\cv'}\\
  & \text{reverse PGM:} & \mapsto & \quad \ket{\phi_5} &=& \quad \dfrac{1}{\sqrt{\abs{\C}}}\sum_{\cv, \cv' \in \C}\alpha_{\cv,\cv'}\ket{\cv-\cv'}\ket{Y_{\cv'}}\ket{\zero}\\
  & & & &=& \quad \sqrt{\PPGM}\ket{\zero}\lp\dfrac{1}{\sqrt{\abs{\C}}}\sum_{\cv \in \C}\ket{Y_\cv}\rp +  \sum_{\cv, \cv'\neq \cv }\alpha_{\cv,\cv'}\ket{\cv-\cv'}\ket{Y_{\cv'}}\\
  & \text{measure 1st register, if $\ket{\zero}$:} & \mapsto & \quad \ket{\phi_6} &= & \quad \dfrac{1}{\sqrt{\abs{\C}}}\sum_{\cv \in \C}\ket{Y_{\cv}} = \dfrac{1}{Z_0}\sum_{\yv \in \C^\perp}\hat{f}(\yv)\ket{\yv} \;\;\text{w.p. $\PPGM$} \\
	&\text{measuring the whole state:} & \mapsto & & &\quad  \ket{\yv}  \;\;\text{with    $\yv \in \C^\perp$.}
\end{align*}}}

\vspace{0.2cm}
As shown in \cite[\S 6.3]{CT23} in the case of the Bernoulli distribution there are choices of the code parameters $k$ and $n$ for which we measure $\ket{\zero}$ with probability $1-o(1)$. This
can be avoided by using a slight tweak of this algorithm as shown in \cite[\S 6.3.2]{CT23}.

\subsection{A tweak in the algorithm avoiding to measure the zero codeword}
The idea is to tweak the algorithm in order to get as the last state before the final measurement the state $\ket{U_0}$ rather than $\ket{\phi_6} = \dfrac{1}{Z_0}\sum_{\yv \in \C^\perp}\hat{f}(\yv)\ket{\yv}$, where 
$$
\ket{U_{0}}   = \frac{\sum_{\yv \in \C^\perp : \yv \neq \mathbf{0}} \QFT{f}(\yv)\ket{\yv}}{\norm{\sum_{\yv \in \C^\perp : \yv \neq \mathbf{0}} \QFT{f}(\yv)\ket{\yv}}}.
$$
In order to achieve this, we replace the PGM by a measurement in a basis defined by the $\ket{Z_\cv}$'s which basically correspond to
the $\ket{Y_\cv}$ where we removed the $\ket{\zero}$ component. In other words, we define for all $\cv \in \C$ the states
$$
\ket{Y'_\cv}   = \frac{1}{\sqrt{q^k}} \left(\ket{U_0} + \sum_{\sv \neq \mathbf{0}} \chi_\cv(u_\sv) \ket{\tW_\sv} \right).
$$
These states are readily seen to be orthogonal. However projecting in this basis does not make the measurement complete. This can be circumvented by adding to the basis the $\ket{\zero}$ state.
We add an extra outcome $\uv \in \F_q^n\backslash \C$, define $\sS = \C \cup \{\uv\}$ and set $\ket{Y'_{\uv}}  = \ket{\zero}$.  We therefore have a projective measurement $\{\ket{Y'_{\yv}}\}_{\yv \in \sS}$.
Again, we have 
$
\braket{\QFT{\psi_\cv}}{{Y'}_\cv} \ge \sqrt{\PPGM} - \frac{1}{\sqrt{q^k}}
$
and independent of $\cv$ so w.p. at least $\left(\sqrt{\PPGM} - \frac{1}{\sqrt{q^k}}\right)^2$, we get the state $\ket{U_0}$. Then, if we measure this state $\ket{U_0}$,  we will get a dual codeword $\yv$ with probability $q(\yv)$ given by

\begin{definition}\label{def: probability}
   We define the probability of measuring the state $\yv \in \CCps$ as $$q(\yv) = \dfrac{\abs{\hat{f}(\yv)}^2}{\sum_{\zv \in \CCps}\abs{\hat{f}(\zv)}^2}$$
\end{definition}

All non zero elements of $\Fqn$ have the same probability of belonging to $\CC^\perp$, namely $q^{-k}$. This suggests that the probability  $q$ should take almost all its mass in $\Typ \cap \CCps$ since $p = |\hat{f}|^2$ takes almost all its mass on $\Typ$. We will show that this holds as long as we are in the achievability region under very mild assumptions on the typical set. 
We are now going to prove that for a family of nice typical sets we have

\begin{theorem}\label{thm:generalPM}
	Let $q \ge 2$ be a prime power. We assume that we have an increasing sequence of positive integers
        $(n_i)_{i \geq 1}$ and a family of functions $(f_i)_{i \geq 1}$ with $f_i : \F_q^{n_i} \rightarrow \mathbb{C}$ with $\norm{f_i}_2 = 1$  such that
        the sequence $(p_i)_{i \geq 1}$ defined by $p_i \eqdef |\QFT{f_i}|^2$ admits a nice family of typical sets not containing $0$. We let
        $H \eqdef \lim_{i \rightarrow \infty}\frac{H_q(p_i)}{n_i}$. Let $R$ be a real number chosen in $(0,H)$.
        Consider a family $(\C_i)_{i \ge 1}$ of random linear codes where $\C_i$ is chosen with a generator matrix chosen uniformly at random in $\F_q^{\lfloor Rn_i\rfloor \times n_i}$.
        Then, for any $\epsilon > 0$ the algorithm of the quantum reduction with PGM using the measurement basis $\{\ket{Y'_{\yv}}\}_{\yv \in \sS}$ described above, outputs codewords of $(\CC_i^\perp)^*$ in the typical set $\Typi$ with probability $1-\negl(i)$ as $i$ tends to infinity.
\end{theorem}

A straightforward corollary of this proposition is the following result
\begin{proposition}\label{prop:measuringtypicalset}
Let $g : \Fq \rightarrow \mathbb{C}$ with $\norm{g}_2 = 1$ and $f : \Fqn \rightarrow \mathbb{C}$ such that $f = g^{\otimes n}$ where $n$ is an arbitrary positive integer.
Let $R$ be a real number in $(0,H_q(\abs{\fg}^2))$. Choose a linear code $\CC$ at random by picking its generator matrix uniformly at random in
$\F_q^{\lfloor Rn \rfloor \times n}$.
Then, for any $\epsilon > 0$ the algorithm of the quantum reduction with PGM using the measurement basis $\{\ket{Y'_{\yv}}\}_{\yv \in \sS}$ described above, outputs codewords of $(\CC^\perp)^*$ in the typical set $\Typ$ with probability $1-\negl(n)$ as $n$ tends to infinity.
\end{proposition}

\begin{proof}
This is an application of Theorem \ref{thm:generalPM} where $f_i = g^{\otimes i}$ and we
take $n_i=i$, $H= H_q(\abs{\QFT{g}}^2)$, $A_i(\epsilon)= q^{-i(H+\eps)}$, $B_i(\epsilon)= q^{-i(H-\eps)}$. Proposition \ref{Proposition:1} shows that
$\delta(\eps,n)=\negl(n)$ for all $\eps >0$. Moreover, we can always exclude $0$ from the typical set, it stays a nice typical set (it satisfies the three conditions of the definition). There is only one case where we can not do this, which corresponds to the case where $\QFT{g}(0)=1$. However in this case $H=0$ and we apply the theorem to the empty set.
\end{proof}

The case of the rank metric is also straightforward.
\begin{proposition}\label{prop:measuringtypicalsetrank}
	Let $q \geqslant 2$ be a prime power. We assume that there exist three increasing sequences $(a_i)_{i \ge 1}$, $(b_i)_{i \ge 1}$, $(t_i)_{i \ge 1}$ such that $a_i \geq b_i$ for every positive integer $i$ and the
		sequence  $(1+t_i/a_i)(1-t_i/b_i)$ tends to a constant $H \in (0,1)$. Let $R$ be a constant in $(0,H)$.
                 Consider a family $(\C_i)_{i \ge 1}$ of random linear codes where $\C_i$ is chosen with a generator matrix chosen uniformly at random in
        $\F_q^{\lfloor Rn_i\rfloor \times n_i}$.
        Then, for any $\epsilon > 0$ the algorithm of the quantum reduction with PGM using the measurement basis $\{\ket{Y'_{\yv}}\}_{\yv \in \sS}$ described above, outputs codewords of $(\CC_i^\perp)^*$ in the typical set $\Typi$ with probability $1-\negl(i)$ as $i$ tends to infinity.
\end{proposition}

\begin{proof}
The $0$ codeword is clearly excluded from the typical set when $\epsilon$ is small enough.
From \eqref{eq:delta} we know that
          $\delta(\epsilon,n_i) = \negl(i)$ for any $\eps >0$.	\eqref{eq:B} shows that we can take $\alpha(\eps)=0$, $\beta(\eps)=2\eps$ and the assumptions of Theorem \ref{thm:generalPM} are verified.         
\end{proof}

\begin{remk} It is insightful that at the limit when $R=H$ and large $i$ we have
$$
R \approx (1+t_i/a_i)(1-t_i/b_i).
$$
The Gilbert-Varshamov distance $t_{\text{dGV}}(a_i,b_i,1-R)$ of rate $1-R$ and length $n_i = a_i b_i$ (which is the rate of the dual code of $\CC_i$)
corresponding to the rank metric viewing vectors of $\F_q^{n_i}$ as matrices in $\F_q^{a_i \times b_i}$ is the largest integer $t$ such that
\begin{equation}
\label{eq:dgvrank}
Ra_ib_i > a_it+t(b_i-t).
\end{equation}
This corresponds to the minimum rank distance of a random linear code of rate $1-R$ over $\F_q^{a_ib_i}$ viewed as matrix code over
$\F_q^{a_i \times b_i}$.
On the other hand since we measure with the PGM elements of the typical set we get at the limit where $\eps=0$ elements of rank weight $b_i-t_i$. It is straightforward to verify that if
$R= (1+t_i/a_i)(1-t_i/b_i)$ then $Ra_ib_i = (a_i+t_i)(b_i-t_i)$ which gives
$Ra_ib_i = a_i(b_i-t_i) + (b_i-t_i)(b_i-(b_i-t_i))$. In other words, when we compare this to \eqref{eq:dgvrank}, $b_i-t_i$ is precisely the minimum rank distance of $\CC_i^\perp$, and we obtain at the limit where the PGM works elements of minimum rank distance. The reason for this behavior will become apparent in the following subsection.
\end{remk}

Let us prove now Theorem \ref{thm:generalPM}.
Before proving this proposition a few lemmas will be very helpful. We let 
\begin{eqnarray*}
X & \eqdef & \sum_{\yv \in \Typi \cap \CCps} p(\yv)\\
Y & \eqdef & \sum_{\yv \in  \CCps}p(\yv).
\end{eqnarray*}
Note that the probability of outputting an element of $\CCps \cap \Typ$  is nothing but the ratio $\frac{X}{Y}$. We expect that $X$ and $Y$ concentrate around their expectation. The following lemmas 
will use similar ideas as in the proof of Theorem~\ref{Theorem:Tractability}. 
In particular we can set $\sv=0$ in Lemma \ref{lem:concentrationZsepsprime} and get 
\begin{lemma}\label{lem:concentrationX}
Let $\gamma \eqdef H - R$. If $\epsilon$ is such that $2\beta(\epsilon)+\alpha(\epsilon) \leq \gamma/2$, then  we have for any positive integer $i$
$$
\Pr_\Gm\lc X \ge \frac{1 -\delta(\epsilon,n_i)- K_2 q^{-n_i(H+\beta(\epsilon))}-q^{-\frac{\gamma n_i}{8}}}{q^{k_i}} \rc = 1 - \negl(i),
$$
where $k_i \eqdef \lfloor Rn_i \rfloor$.
\end{lemma}

On the other hand the tails of $Y$ above its expectation are bounded by
\begin{lemma}
\label{lem:concentrationY}
 Let $\gamma \eqdef H-R$. We choose $\epsilon$ a positive constant such that $2\beta(\epsilon)+\alpha(\epsilon) \leq \gamma/2$.
We have $\Pr\lp Y/X \leq 1 + \sqrt{\delta(\eps,n_i)} \rp = 1 - \negl(i)$.
\end{lemma}

\begin{proof}
	
It will be helpful to decompose $Y$ as:
\begin{eqnarray*}
Y & = & \sum_{\yv \in \CCps} p(\yv)\\
& = & \underbrace{\sum_{\yv \in \CCps \cap \Typi} p(\yv)}_{X} +  \underbrace{\sum_{\yv \in \CCps \setminus \Typi} p(\yv)}_{\eqdef Z}.
\end{eqnarray*}
This implies that
\begin{equation}
\label{eq:YbyX}
\Pr\lp Y/X \leq 1 + \sqrt{\delta(\eps,n_i)} \rp = \Pr \lp Z/X \le \sqrt{\delta(\eps,n_i)} \rp. 
\end{equation}
We define now the following events
\begin{eqnarray*}
\sE_1 & \eqdef & \left\{ X \geq \frac{1 - \delta(\eps,n_i)-K_2 q^{-n_i(H+\beta(\eps))}-q^{-\frac{\gamma n_i}{8}}}{q^{k_i}} \right\}\\
\sE_2 & \eqdef & \left\{ Z \leq \frac{\sqrt{\delta(\eps,n_i)}(1 - \delta(\eps,n_i)-K_2 q^{-n_i(H+\beta(\eps))}-q^{-\frac{\gamma n_i}{8}})}{q^{k_i}} \right\}.
\end{eqnarray*}
We have
\begin{equation}
\label{eq:intersection}
\Pr \lp Z/X \le \sqrt{\delta(\eps,n_i)} \rp \ge  \Pr(\sE_1 \cap \sE_2) = \Pr(\sE_1) + \Pr(\sE_2) - \Pr(\sE_1 \cup \sE_2) \ge \Pr(\sE_1) + \Pr(\sE_2) - 1.
\end{equation}
The first term is handled by Lemma \ref{lem:concentrationX}, we have
\begin{eqnarray}
\Pr\lp X \geq \frac{1 - \delta(\eps,n_i)-K_2 q^{-n_i(H+\beta(\eps))}-q^{-\frac{\gamma n_i}{8}}}{q^{k_i}}\rp & = & 1 - \negl(i). \label{eq:E1}
\end{eqnarray}
The second probability $\Pr(\sE_2)$ is treated by the Markov inequality. We 
observe that
\begin{eqnarray}
	\esp\{Z\} &= &\esp \left\{\sum_{\yv \in \CCps \setminus \Typi}p(\yv) \right\} \nonumber\\
		& \leq &  \sum_{\yv \notin \Typi}p(\yv) \Pr(\yv \in \CCps )\nonumber \\
		&= &\dfrac{1}{q^{k_i}}\sum_{\yv \notin \Typi}p(\yv)\nonumber \\
                &= &\dfrac{\delta(\epsilon,n_i)}{q^{k_i}}. \label{eq:EZ}
\end{eqnarray}
Let $p \eqdef \delta(\eps,n_i)+K_2 q^{-n_i(H+\beta(\eps))}+q^{-\frac{\gamma n_i}{8}}$. By using Markov's inequality we get
\begin{eqnarray}
\Pr \lp Z \geq \frac{\sqrt{\delta(\eps,n_i)}(1- p)}{q^{k_i}} \rp & \leq & \frac{\esp\{Z\}}{\tfrac{\sqrt{\delta(\epsilon, n_i)}(1-p)}{q^{k_i}}} \label{eq:Zmarkov}\\
& \leq & \frac{\sqrt{\delta(\epsilon,n_i)}}{1-p}\label{eq:ZespZ}\\
& = & \negl(i), \label{eq:prop1}
\end{eqnarray}
where \eqref{eq:Zmarkov} follows directly from Markov's inequality, \eqref{eq:ZespZ} follows from the upper bound \eqref{eq:EZ} on the expectation of $Z$ and
\eqref{eq:prop1} follows from the fact that family of typical sets is nice.
We conclude the proof by using \eqref{eq:E1} and \eqref{eq:prop1} in \eqref{eq:intersection}.
\end{proof}

By using these two lemmas we obtain directly Theorem \ref{thm:generalPM} as shown by
\begin{proof}[Proof of Theorem \ref{thm:generalPM}]
Let $\epsilon$ be any positive real such that $2\beta(\epsilon)+\alpha(\epsilon) \leq \gamma/2$. The probability of measuring an element of $\Typi$ is given by the ratio $X/Y$.  By using Lemma \ref{lem:concentrationY} we know that
$X/Y$ is greater than $\frac{1}{1+\sqrt{\delta(\epsilon,n_i)}}$ with probability $1 - \negl(i)$ over the choice of the generator matrix $\Gm$.
This leads to the proof of the theorem by observing that 
(i) the probability of measuring an element of $\Typi$ is increasing with $\epsilon$, therefore if this property holds for a given $\epsilon_0 >0$, it holds for any $\eps  \geq \eps_0$,
(ii) we can choose $\eps$ arbitrarily close to $0$ satisfying $2\beta(\epsilon)+\alpha(\epsilon) \leq \gamma/2$.
\end{proof}

In other words, we obtain with this quantum reduction algorithm a way of obtaining elements  in $\Typi$ whose probability is in the range
$\left[q^{-n_i(H+\alpha(\epsilon))},q^{-n(H-\beta(\epsilon))}\right]$. It turns out that when the difference $H-R$ is small we can not have much better (namely elements of bigger probability) as shown by
\begin{proposition}\label{lem:maximal}
	Let $q \ge 2$ be a prime power, $R$ be a fixed real number in $(0,1)$. We choose a
        random code $\C$ in $\F_q^n$ by choosing a generator matrix in $\Fq^{\lfloor Rn \rfloor \times n}$ uniformly at random.
        Let $p$ be a probability distribution on $\F_q^n$. Let $\epsilon$ be a fixed positive real number.
        The probability over $\Gm$ that there exists a non zero codeword $\cv$ in $\CCp$ such that $p(\cv) > q^{-n(R-\epsilon)}$ is
a negligible function of $n$.
\end{proposition}
\begin{proof}
For $\epsilon>0$ let $\Bc_\epsilon^{(n)} \eqdef \left\{\yv \in \left(\F_q^n\right)^*,\;p(\yv)>q^{-n(R-\epsilon)}\right\}$. Let $k \eqdef \lfloor Rn \rfloor$.
From 
Lemma \ref{Variance}, we have $\esp\lp\abs{ \bTyp \cap \CCps}\rp = \dfrac{ \abs{  \bTyp }}{q^{k}}$. 
Since for any $\epsilon >0$, $p( \bTyp) \leq 1$ and $p( \bTyp)  >\abs{  \bTyp } q^{-n(R-\epsilon)}$ by definition of $\bTyp$ we deduce that $ \abs{  \bTyp } < q^{n(R-\epsilon)}$. Therefore 
\begin{eqnarray*}
\esp\lp\abs{ \bTyp \cap \CCps}\rp& < & \frac{q^{n(R-\epsilon})}{q^{k}}\\
& = & q^{-\eps n + \OO{1}} \\
& = & \negl(n).
\end{eqnarray*}
Since $\Pr\lp \bTyp \cap \CCps \neq \emptyset \rp \leq \esp\lp\abs{ \bTyp \cap \CCps}\rp$ we deduce that 
$\Pr\lp \bTyp \cap \CCps \neq \emptyset \rp = \negl(n)$.
\end{proof}

This proposition together with Proposition \ref{prop:measuringtypicalset} leads for instance in the memoryless case to the following theorem.

\reductionpgm*

\subsection{Distribution associated to a metric}
	We study an application of the previous results when the probability distribution is associated to an additive metric on $\Fqn$. What we mean here is a metric defined from a metric $\abs{\cdot}_{\Fq}$ over $\Fq$ as $\abs{\xv}_{\Fqn} = \sum\limits_{i = 1}^n \abs{x_i}_{\Fq}$. The distribution $p$ associated to this metric on $\Fqn$ is defined by  $p(\xv) = q^{-\lambda \abs{\xv}_{\Fqn} }$ for a unique $\lambda > 0$.   Let $g : \Fq \rightarrow \mathbb{C}$ with $\norm{g}_2 = 1$ such that $f = g^{\otimes n}$. We recall that $p = \abs{\hf}^2 $ and we define $r  = \abs{\fg}^2$. 
	\begin{lemma}\label{lemma : p}
		For each metric on $\Fq$, there exists a unique $\lambda  > 0$ such that  $r(\alpha)$ defined by $r(\alpha) = q^{-\lambda \abs{\alpha}_{\Fq} } $ is a probability distribution .\\
	\end{lemma}
	\begin{proof}
		Let $F(\lambda) = \sum\limits_{\alpha \in \Fq} q^{-\lambda \abs{\alpha}_{\Fq} }$. 
		Observe that $F(0) = \sum\limits_{\alpha \in \Fq}  q^{- 0 \abs{\alpha}_{\Fq} } = q$ and  $\lim\limits_{\lambda \rightarrow \infty}F(\lambda) = 0$. Since $F$ is continuous and decreasing over $[ 0, \infty [$, we deduce that there exists a unique $\lambda > 0$ such that $F(\lambda) = 1$, which finishes the proof.

	\end{proof}
	This probability distribution is a product probability distribution $p^\perp(\xv) = \prod\limits_{i = 1}^n r(x_i)$ with $r(x_i) = q^{-\lambda \abs{x_i}_{\Fq} }$. To apply Algorithm \ref{algo: reductionPGM}, we can construct the amplitude function $f$ defined on $\Fqn$ such that $\abs{\QFT{f}}^2 = p$ by choosing $\QFT{f}= \sqrt{p} $ which leads to $f(\xv) = \dfrac{1}{\sqrt{q^n}} \sum_{\yv \in \Fqn} \chi_\yv(-\xv) \sqrt{p(\yv)}$.
	
Using Theorem \ref{thm : reductionpgm}, it turns out that when $R$ approaches its limit $H_q(\abs{\fg}^2)$ at which we might hope to decode, Algorithm \ref{algo: reductionPGM} outputs non-zero codewords in $\CCp$ of probability larger than $q^{-nH_q(\abs{\fg}^2)}$ with probability $1-o(1)$ and there exists no non-zero codeword in $\CCp$ of smaller weight. Then, as the distribution  $p$ is a decreasing function of $ \abs{\xv}_{\Fqn} $, we obtain non-zero codewords of nearly minimum weight.

\begin{theorem}
	Let $g : \Fq \rightarrow \mathbb{C}$ with $\norm{g}_2 = 1$ such that $f = g^{\otimes n}$. Let $\C$ be a linear code on $\Fqn$. For any $\epsilon > 0$:
	\begin{itemize}
		\item We assume that there exists $\gamma > 0$ such that  $H_q(\abs{\fg}^2) - R >\gamma  $. Algorithm \ref{algo: reductionPGM} outputs non-zero codewords $\xv \in \CCp$ such that $\abs{\xv}_{\Fqn} \geqslant \frac{n(H_q(\abs{\fg}^2)+\epsilon)}{\lambda}$ with probability $1-o(1)$ as $n$ tends to infinity.
		\item The probability that there exists a non-zero codeword $\xv$ in $\CCp$ such that $\abs{\xv}_{\Fqn} \geqslant \frac{n(R - \epsilon)}{\lambda}$ is $o(1)$.
	\end{itemize}

\end{theorem}

\begin{proof}[Proof of the Theorem]
	Using Theorem \ref{thm : reductionpgm}, in the first part, we know that if there exists $\gamma > 0$ such that  $H_q(\abs{\fg}^2) - R >\gamma  $, then for any $\epsilon > 0$, Algorithm \ref{algo: reductionPGM} outputs non-zero codewords $\xv$ in $\CCp$ of probability larger than $q^{-n(H_q(\abs{\fg}^2)+\epsilon)}$ with probability $1-o(1)$ as $n$ tends to infinity. Here, since $p(\xv) = q^{-\lambda \abs{\xv}_{\Fqn} }$, we get vectors $\xv \in \Fqn$ such that $\abs{\xv}_{\Fqn} \geqslant \frac{n(H_q(\abs{\fg}^2)+\epsilon)}{\lambda}$ with probability $1-o(1)$.
	
	Then, the second part claims that for any $\epsilon ' > 0$,  the probability over $\Gm$ that there exists a non-zero codeword in $\CCp$ of probability larger than $q^{-n(R-\epsilon')}$ goes to $0$ as $n$ tends to infinity. Here, we get that the probability that there exists no nonzero codeword $\xv$ in $\CCp$ such that $\abs{\xv}_{\Fqn} \geqslant \frac{n(R - \epsilon')}{\lambda}$ is $1-o(1)$.
	
\end{proof}

\subsection{Solving the decoding problem in the surjective regime} \label{ss:surjective_decoding}\label{ss:Renes}
As explained in the introduction a simple variation of Regev's procedure allows to solve the decoding problem in the surjective regime. PGM can be used as for finding low weight codewords as follows

\noindent\fbox{\parbox{\textwidth}{
\begin{center}
{\bf Algorithm 4.4 : Solving the decoding problem in the surjective regime with PGM.}\label{algo: reductionPGM}
\end{center}
\noindent
{\bf Input:} $\CC \subset \F_q^n$, $\vv \in \F_q^n$\\
{\bf Output:} an element of $\cv$ such that $p(\cv - \vv)$ is large.
\begin{align*}
	&\text{Initialization:} &  & \quad \ket{\phi_0} &=& \quad \dfrac{1}{\sqrt{\abs{\C}}}\sum_{\cv \in \C}\sum_{\ev \in \F_q^n}\chi_\cv(-\vv)f(\ev)\ket{\cv}\ket{\ev}\ket{\zero} \\
  &\text{adding $\cv$ to $\ev$:} & \mapsto &  \quad \ket{\phi_1} &=& \quad \frac{1}{\sqrt{|\C|}} \sum_{\cv \in \C} \chi_\cv(-\vv) \ket{\cv} \sum_{\ev \in \F_q^n}   f(\ev) \ket{\cv+\ev}\ket{\zero}\\
  & & & &=& \quad \frac{1}{\sqrt{|\C|}} \sum_{\cv \in \C} \chi_\cv(-\vv) \ket{\cv} \ket{\psi_\cv}\ket{\zero}  \\
        & \text{apply the QFT on $\ket{\psi_\cv}$:} & \mapsto & \quad \ket{\phi_2} &=& \quad \dfrac{1}{\sqrt{\abs{\C}}}\sum_{\cv \in \C} \chi_\cv(-\vv) \ket{\cv}\ket{\QFT{\psi_\cv}}\ket{\zero} \\
  &\text{apply PGM coherently:} & \mapsto  & \quad
  \ket{\phi_3}& =& \quad \dfrac{1}{\sqrt{\abs{\C}}}\sum_{\cv \in \C} \chi_\cv(-\vv) \ket{\cv}\sum_{\cv' \in \C} \alpha_{\cv,\cv'}\ket{Y_{\cv'}}\ket{\cv'}\\
  & & & & &\text{ with } \ket{\QFT{\psi_\cv}} = \sum_{\cv' \in \C}\alpha_{\cv,\cv'}\ket{Y_{\cv'}} \\
  &\text{erase $\cv$:} & \mapsto & \quad \ket{\phi_4} &=& \quad \dfrac{1}{\sqrt{\abs{\C}}}\sum_{\cv, \cv' \in \C}\chi_\cv(-\vv) \alpha_{\cv,\cv'}\ket{\cv-\cv'}\ket{Y_{\cv'}}\ket{\cv'}\\
  & \text{reverse PGM:} & \mapsto & \quad \ket{\phi_5} &=& \quad \dfrac{1}{\sqrt{\abs{\C}}}\sum_{\cv, \cv' \in \C} \chi_\cv(-\vv) \alpha_{\cv,\cv'}\ket{\cv-\cv'}\ket{Y_{\cv'}}\ket{\zero}\\
  & & & &=& \quad \sqrt{\PPGM}\ket{\zero}\lp\dfrac{1}{\sqrt{\abs{\C}}}\sum_{\cv \in \C}\chi_{\cv}(-\vv)\ket{Y_\cv}\rp +  \sum_{\cv, \cv'\neq \cv }\alpha_{\cv,\cv'}\ket{\cv-\cv'}\chi_{\cv}(-\vv)\ket{Y_{\cv'}}\\
  & \text{meas. 1st register, if $\ket{\zero}$:} & \mapsto & \quad \ket{\phi_6} &= & \quad \dfrac{1}{\sqrt{\abs{\C}}}\sum_{\cv \in \C}\chi_{\cv}(-\vv)\ket{Y_{\cv}} = \dfrac{1}{Z_{\vv \trsp{\Gm}}}\sum_{\yv \in \vv+\C^\perp}\hat{f}(\yv)\ket{\yv} \;\;\text{w.p. $\PPGM$} \\
	&\text{meas. the whole state:} & \mapsto & & &\quad  \ket{\yv}  \;\;\text{with    $\yv \in \vv+\C^\perp$.}
\end{align*}}}

The penultimate step (that is the first measuring step) can be justified by the following computation:
\begin{eqnarray*}
\dfrac{1}{\sqrt{\abs{\C}}}\sum_{\cv \in \C}\chi_{\cv}(-\vv)\ket{Y_{\cv}} &= & \dfrac{1}{q^k \sqrt{\abs{\C}}}\sum_{\cv \in \C}\chi_{\cv}(-\vv) \sum_{\sv \in \F_q^k} \chi_\cv(\uv_\sv) \ket{\tilde{W_\sv}}
\quad \text{(by \eqref{eq:Yc})}\\
& = & \dfrac{1}{q^k \sqrt{\abs{\C}}} \sum_{\sv \in \F_q^k} \ket{\tilde{W_\sv}}\sum_{\cv \in \C}\chi_\cv(\uv_\sv-\vv).
\end{eqnarray*}
It is readily verified that $\sum_{\cv \in \C}\chi_\cv(\uv_\sv-\vv)=0$ for all $\sv \in \F_q^k$ with the exception of $\trsp{\sv} = \Gm \trsp{\vv}$, where
$\sum_{\cv \in \C}\chi_\cv(\uv_\sv-\vv)=q^k$. This implies that 
\begin{eqnarray*}
\dfrac{1}{\sqrt{\abs{\C}}}\sum_{\cv \in \C}\chi_{\cv}(-\vv)\ket{Y_{\cv}} 
& = & \dfrac{1}{\sqrt{\abs{\C}}} \ket{\widetilde{W_{\vv \trsp{\Gm}}}}.
\end{eqnarray*}
By definition of $\ket{\widetilde{W_{\vv \trsp{\Gm}}}}$ we have
$$
\ket{\widetilde{W_{\vv \trsp{\Gm}}}} = \frac{1}{Z_{\vv \trsp{\Gm}}} \sum_{\yv \in \vv + \C^\perp} \hat{f}(\yv)\ket{\yv},
$$
and this justifies that we get $\dfrac{1}{Z_{\vv \trsp{\Gm}}}\sum_{\yv \in \vv+\C^\perp}\hat{f}(\yv)\ket{\yv}$ after measuring $\zero$ in the first register. 

We have treated in the previous subsections the case $\sv=\zero$ where we want to avoid to measure the trivial zero solution. This is why we used the projective measurement with respect to the $\ketbra{Y'_\yv}{Y'_\yv}$'s where $\yv$ ranges over $\sS$. It is readily seen that Algorithm 4.4 can be analyzed in the same way we analyzed the modified version of Algorithm 4.1 based on the projective measurement with respect to the $\ketbra{Y'_\yv}{Y'_\yv}$'s. We replace in all the analysis $(\C^\perp)^*$ with $\vv + \C^\perp$ and obtain readily a version of Theorem \ref{thm:generalPM}, Proposition \ref{prop:measuringtypicalset}, Proposition \ref{prop:measuringtypicalsetrank} where Algorithm 4.4 outputs elements of $\vv+ \CC_i^\perp$ in the typical set $\Typi$ with probability $1-\negl(i)$ as $i$ tends to infinity.
There are also corresponding non achievability results, for instance a modified version of Proposition \ref{lem:maximal} showing that the  probability over $\Gm$ that there exists an element $\yv$ in $\vv + \CCp$ such that $p(\yv) > q^{-n(R-\epsilon)}$ is
a negligible function of $n$. Again, this shows that at the limit where PGM works we basically get elements of $\vv + \CCp$ of maximal probability $p$ with Algorithm 4.4. This gives in essence an operational version of \cite[\S 5, cor. 2]{R18a}. We namely  prove with this approach
\reductionpgmwithv*

	\newcommand{\etalchar}[1]{$^{#1}$}
	
	\bibliographystyle{alpha}
\newpage
\end{document}